\documentclass[reqno,12pt,letterpaper]{amsart}
\usepackage{amsmath,amssymb,amsthm,graphicx,mathrsfs,url}
\usepackage[usenames,dvipsnames]{color}
\usepackage[colorlinks=true,linkcolor=Red,citecolor=Green]{hyperref}
\usepackage{amsxtra}
\usepackage{subcaption}
\usepackage{verbatim}
\usepackage{lipsum}
\usepackage{mwe}

\def\?[#1]{\textbf{[#1]}\marginpar{\Large{\textbf{??}}}} 

\let\epsilon=\varepsilon 

\newcommand{\RR}{{\mathbb R}}
\newcommand{\NN}{{\mathbb N}}
\newcommand{\CC}{{\mathbb C}}
\newcommand{\diag}{\mathrm{diag}}

\newcommand{\ZZ}{{\mathbb Z}}
\newcommand{\bt}{\mathbf{t}}
\newcommand{\Hn}{H_n ( \alpha; \textbf{t})}
\newcommand{\Hkn}{H_{n,\mathbf{k}} ( \alpha; \textbf{t})}
\newcommand{\Dn}{D_n ( \alpha; \textbf{t})}
\newcommand{\kk}{\mathbf{k}}
\newcommand{\bz}{\bar{z}}
\newcommand{\Dnn}{D_n ( 0; \textbf{t})}
\newcommand{\bo}{\mathbf{0}}
\newcommand{\bi}{\mathbf{i}}

\newcommand{\be}{\mathbf{e}}

\newtheorem{theo}{Theorem}
\newtheorem{prop}{Proposition}[section]

\newtheorem{lemm}[prop]{Lemma}

\newtheorem{rem}{Remark}
\numberwithin{equation}{section}

\DeclareMathOperator{\Spec}{Spec}

\let\Im=\Imag

\DeclareMathOperator{\vol}{vol}

\DeclareMathOperator{\id}{\text{Id}}

\usepackage{scalerel}

\newcommand\reallywidehat[1]{\arraycolsep=0pt\relax%
\begin{array}{c}
\stretchto{
  \scaleto{
    \scalerel*[\widthof{\ensuremath{#1}}]{\kern-.5pt\bigwedge\kern-.5pt}
    {\rule[-\textheight/2]{1ex}{\textheight}} 
  }{\textheight} %
}{0.5ex}\\           
#1\\                 
\rule{-1ex}{0ex}
\end{array}
}
\title[Twisted multilayer graphene]{Flat Bands and High Chern Numbers in Twisted Multilayer Graphene} 

\author{Mengxuan Yang}
\email{mxyang@math.berkeley.edu}
\address{Department of Mathematics, University of California, Berkeley, CA 94720, USA.}

\begin{document}

\begin{abstract}
Motivated by recent Physical Review Letters of Wang--Liu \cite{wang2022hierarchy} and Ledwith--Vishwanath--Khalaf  \cite{ledwith2022family}, we study Tarnopolsky--Kruchkov--Vishwanath \cite{tarnopolsky2019origin} chiral model of two sheets of $n$-layer Bernal stacked graphene twisted by a small angle using the framework developed by Becker--Embree--Wittsten--Zworski \cite{becker2020mathematics}. We show that magic angles of this model are exactly the same as magic angles of chiral twisted bilayer graphene with multiplicity. For small inter-layer tunneling potentials, we compute the band separation at Dirac points as we turning on the tunneling parameter. Flat band eigenfunctions are also constructed using a new theta function argument and this yields a complex line bundle with the Chern number $-n$. 
\end{abstract}

\maketitle

\section{Introduction}
When two or more sheets of graphene are stacked on top of each other and twisted, it has been observed that at certain angles, namely \emph{magic angles}, the zero energy band becomes flat and the system becomes superconducting. 

In this paper, we consider the chiral model \cite{san2012non} \cite{tarnopolsky2019origin} of two sheets of $n$-layer Bernal stacked graphene twisted by a small angle \cite{wang2022hierarchy}. The twisted multilayer graphene (TMG) Hamiltonian is given by 
\begin{equation}
  \label{eq:defH}    
  H_n ( \alpha; \bt)  := \begin{pmatrix} 0 & D^*_n(\alpha; \mathbf{t}) \\
  D_n ( \alpha; \mathbf{t}) &  0 \end{pmatrix}    
\end{equation} 
where 
\begin{equation}
  \label{eq:defD}
  D_n ( \alpha; \bt ) := \left(\begin{matrix} D (\alpha) & t_1T_+ & \\ t_1T_- & D(0) & t_2T_+ & \\ & t_2T_- & D(0) & \ddots \\ && \ddots && \\ &&&& t_{n-1}T_+ \\ &&&t_{n-1}T_- & D(0)\end{matrix}\right),
\end{equation}
with $\bt=(t_1, t_2, \cdots, t_{n-1})$ and 
\begin{equation}
  \label{eq:defD1}
  D( \alpha ) := \begin{pmatrix} { 2 }  D_{\bar z}  &   \alpha U(z) \\  
    \alpha U(-z) & { 2 } D_{\bar z }   \end{pmatrix}, \ \ 
    T_+ = \left(\begin{matrix}1&0\\0&0\end{matrix}\right), \ \ 
    T_- = \left(\begin{matrix}0&0\\0&1\end{matrix}\right).
\end{equation}
In particular, $ z = x_1 + i x_2$, $ D_{\bar z } := \tfrac1{2 i } ( \partial_{x_1} + 
i \partial_{x_2} )$ and 
\begin{equation}
\label{eq:defU} 
U (z) = U (z,\bar z ) := \sum_{ k = 0}^2 \omega^k e^{ \frac12 (  z  
\bar \omega^k - \bar z  \omega^k )} , \ \
\omega := e^{ 2 \pi i /3 }. 
\end{equation}
In equation \eqref{eq:defD}, $t_kT_+$ (resp.\ $t_kT_-$) denotes the tunneling between the top (resp.\ the bottom) $k$-th layer and $(k+1)$-th layer. Without loss of generality we can assume that all $t_i$'s are non-zero, since otherwise we have direct sum decomposition of $\Dn$ into smaller matrix blocks.

In the corresponding physics model, when two honeycomb lattices are twisted against each another, a periodic honeycomb superlattice, called the moir\'e lattice, is formed. Bistritzer--MacDonald \cite{bistritzer2011moire} predicted that the symmetries of the periodic moir\'e lattice in twisted bilayer graphene (TBG) lead to dramatic flattening of the band spectrum. The chiral model of TBG \eqref{eq:H2} was obtained by Tarnopolsky--Kruchkov--Vishwanath \cite{tarnopolsky2019origin} by removing certain interaction terms from the operator constructed in \cite{bistritzer2011moire}; it was recently studied in greater mathematical details by the work of Becker \emph{et al.} \cite{becker2020mathematics} and Becker--Humbert--Zworski \cite{becker2022fine} \cite{becker2022integrability}. Recently, Wang--Liu \cite{wang2022hierarchy} and Ledwith--Vishwanath--Khalaf \cite{ledwith2022family} generalized the chiral TBG model to chiral TMG models which also have the ideal quantum geometry and support high Chern number band, 
\begin{rem}
For $n=1$, the operator \eqref{eq:defH} reduces to the chiral twisted bilayer (TBG) model
\begin{equation}
\label{eq:H2}
\Hn = H(\alpha) = 
\begin{pmatrix} 0 & D^*(\alpha) \\
  D ( \alpha) &  0 
  \end{pmatrix}
\end{equation}
independent of $t$, see \cite{tarnopolsky2019origin} \cite{becker2020mathematics} \cite{becker2022fine} \cite{becker2022integrability} for more details. The dimensionless parameter $\alpha$ in \eqref{eq:defH} and \eqref{eq:H2} is essentially the reciprocal of the twisting angle.
\end{rem}

Note that the potential \eqref{eq:defU} satisfies the following properties:
\begin{equation}
\label{eq:symmU} 
\begin{gathered}  
U ( z + \mathbf a ) = \bar \omega^{a_1 + a_2}  U ( z ) , \ \ 
\mathbf a = \tfrac{4}3 \pi i \omega (  a_1 + \omega a_2 ) ,  \ \ a_j \in \mathbb Z , 
\\ U ( \omega z )  = \omega U ( z )   , \ \  \overline { U ( \bar z ) } = U ( z ) . 
\end{gathered}
\end{equation}
In particular, it is periodic with respect to the lattice $3\Lambda$ with \begin{equation}
\Lambda := \{ \tfrac{4}{3} \pi i \omega( a_1 + \omega a_2 ) : a_1,a_2 \in \mathbb Z  \}.
\end{equation}
The dual lattice $\tfrac{1}{3}\Lambda^*$ of $3\Lambda$ consists of $ \mathbf k \in\CC $ satisfying
\begin{gather} \langle \gamma,\kk \rangle := \tfrac12 (\gamma \bar{\mathbf k } + \bar \gamma \mathbf k ) \in 2 \pi \ZZ , \ \  \gamma\in 3\Lambda \Longrightarrow
 \Lambda^* = {\sqrt{3}} \omega \left(
  \ZZ  \oplus  {\omega}  \ZZ \right).
\end{gather}
We consider a generalized Floquet condition under the modified translation operator $\mathscr L_{\mathbf{a}}$, with $\mathbf{a}\in \Lambda$ (defined in equation \eqref{eq:defLa}), i.e., we study the spectrum of $ \Hn $ satisfying the following boundary conditions:
\begin{equation}
\label{eq:FL_ev}  
\Hn \mathbf u = E(\alpha,\kk) \mathbf u , \ \  \mathscr L_{\mathbf a } \mathbf u ( z ) = e^{ i \langle  \mathbf{a}, \mathbf k 
\rangle}  \mathbf u ( z ) , \ \ \mathbf k \in \mathbb C . 
\end{equation}
The spectrum is symmetric with respect to the origin due to chiral symmetry and we index it as
\begin{equation}
\label{eq:eigs} 
\begin{gathered} \{ E_{ j } ( \alpha, \mathbf k ) \}_{ j \in  \mathbb Z^* } ,  \ \  E_{j } ( \alpha, \mathbf k ) = - E_{-j} ( \alpha , \mathbf k ) ,\ \ \mathbb Z^*:= \ZZ\backslash \{0\} 
\\ 0 \leq E_1 ( \alpha, \mathbf k ) \leq E_2 ( \alpha, \mathbf k ) \leq \cdots , \ \  E_1 ( \alpha, \mathbf 0 ) =  E_1 ( \alpha, -\mathbf i ) = 0 , 
\end{gathered} \end{equation}
see \S \ref{sec:flo2} for details. The points $ \mathbf 0 , -\mathbf i $
are called the {\em Dirac points} and are typically denoted by $ K $ and $ K' $ in the physics literature. 

For $E( \alpha , \mathbf k )$ satisfies equation \eqref{eq:FL_ev} of $\Hn$, we define the following set of {\em complex numbers}:
\begin{equation}
\label{eq:defA}  \mathcal A_{n} := \{ \alpha \in \mathbb C : \forall \, \mathbf k \in \mathbb C, \ \
E_1 ( \alpha, \mathbf k ) \equiv 0  \} 
\end{equation}
The {\em magic angles} are essentially the reciprocals of $ \alpha$'s. And $\alpha\in \mathcal{A}$ is call {\em simple} if $E_2 ( \alpha, \mathbf k ) >0$ and {\em multiplicity} of $\alpha\in \mathcal{A}$ is defined to be the number of $j$'s ($j>0$) such that $E_j ( \alpha, \mathbf k ) \equiv 0$.

Our first theorem states that the magic angles of the TMG Hamiltonian $\Hn$ coincides with magic angles of the TBG Hamiltonian $H(\alpha)$.
\begin{theo}
\label{thm:1}
The set $\mathcal{A}_n$ is independent of $n\in\NN_+$. More precisely, the multiplicity of each $\alpha\in \mathcal{A}_n$ is independent of $n$.
\end{theo}
Therefore we can define $\mathcal{A}:= \mathcal{A}_n$. This result is essentially contained in the paper \cite{wang2022hierarchy}. A spectral theory characterization of the set $\mathcal{A}_1$ (for TBG) was given in \cite{becker2020mathematics} using a Birman--Schwinger operator whose spectrum is given by $1/\mathcal{A}_1$. 

Our second result studies the role of the tunneling parameter $\bt$ for $\alpha\in\mathcal{A}$. When $\bt=\mathbf{0}$, the operator $\Dn$ can be decomposed into direct sums of one operator $D(\alpha)$ and $(n-1)$ operators $D(0)$, where the latter give rise to Dirac cones independent of $\alpha\in\CC$. The Dirac cones are connected to the band $E_1(\alpha,\kk)$ for any $\alpha\in\CC$. In particular, for $\alpha\in\mathcal{A}$, turning on the tunneling parameter $\bt$ such that $t_i> 0$ for $1\leq i\leq n-1$ results in these Dirac cones separate from the flat band and a band gap is formed. The band separation is studied for $H_{2}(\alpha;t)$: 
\begin{theo}
\label{thm:2}
Assume $\alpha\in\mathcal{A}$ is simple. For $t\ll 1$ and at $\kk=K'$, the first two eigenvalues of $H_{2}(\alpha;t)$ are given by $E_1(\alpha,\kk,\bt)=0$ and $E_2(\alpha,\kk,\bt)=Ct +\mathcal{O}(t^{3/2})$.
\end{theo}
\begin{rem}
The constant $C$ is computed in Proposition \ref{prop:sep}. The same result is expect to be true for $\Hn$. We only present the case $n=2$ here for notatinal simplicity purposes. This also gives an explicit result for the band separation mechanism discussed in \cite{wang2022hierarchy}.
\end{rem}

Using a new theta function argument and the resolvent formula for $(2 D_{\bar{z}}-\kk)^{-1}$, for $\alpha\in\mathcal{A}$ simple, we construct flat band eigenfunctions of $\Hn$. This gives rise to a holomorphic line bundle $L$ (see discussions in \S \ref{sec:chern} for more details).
\begin{theo}
\label{thm:3}
The Chern number of the line bundle L is given by $c_1(L)=-n$.
\end{theo}

This justifies the high Chern number observation made in \cite{wang2022hierarchy} and it relies on a numerical calculation of the integration of the flat band eigenfucntions (see Table \ref{tab:int}). In \cite{wang2022origin}, the flat band eigenfunctions are considered as a section of rank $n$ holomorphic vector bundle. Our new theta function argument gives an explicit analytic construction and answers a question proposed in \cite{liu2022recent}.

We conclude our introduction by discussing its relation to physics:

\noindent
\textbf{Integer quantum Hall effect.}
For two dimensional electron gas in cold temperature, Hall conductance exhibits some plateaus at integer multiples of $e^2/2\pi \hbar$ as the magnetic field varies. It turns out that the Hall conductance at the plateaus is related (via the Berry curvature) to Chern number $c_1$ of the underlying Hermitian line bundles arising from ground state eigenfunctions. By Kubo's formula, the Hall conductance is given by
$$\sigma = -\frac{e^2}{2\pi\hbar}c_1$$
The higher Chern number in the TMG model can yield higher Hall conductance in this case.

\noindent
\textbf{Ideal Quantum Geometry and Fractional Chern insulators.}
Another observation (cf.\,\cite{wang2022hierarchy}) is that flat band eigenfunctions constructed in \eqref{eq:Dn_eigenf} are holomorphic in $\kk\in\CC$. Such flat band is said to satisfy the \emph{ideal quantum geometry} in physics literature (cf.\,\cite{claassen2015position}\cite{ozawa2021relations}\cite{mera2021kahler} \cite{wang2021exact}), in the sense that the Berry curvature $\Omega(\kk)$ on the fundamental domain $\CC/\Lambda^*$ is non-vanishing and proportional to a Fubini--Study metric $g_{\kk}$ on $\CC/\Lambda^*$ induced by the flat band eigenfunctions.

Such flat band was predicted and verified to have Fractional Chern Insulators (FCI) \cite{PhysRevX.1.021014} \cite{bergholtz2013topological}, which are lattice generalizations of the conventional fractional quantum Hall effect in two-dimensional electron gases. For example, FCI can appear at band fillings given by 
$$\nu=\frac{k}{(m-1)|c_1|+1}$$
with $m,k\in\NN$ such that $k\geq1, m\geq2$. Here $k=1$ (resp.\,$k>1$) corresponds to Abelian states (resp.\,non-Abelian states) and $m$ is even (resp.\,odd) for bosons (resp.\,fermions). The stability of FCI are conjectured due to the ideal quantum geometry (cf.\,\cite{tarnopolsky2019origin} \cite{wang2022origin} \cite{liu2022recent}).


\subsection*{Structure of the paper} In Section \ref{sec:TML}, we study symmetry groups and irreducible representations. This yields the so-called protected zero eigenstates for all $\alpha\in\CC$. In Section \ref{sec:floquet}, we recall a generalized Bloch--Floquet theory introduced in \cite{becker2022fine} and give spectral characterizations of magic angles. In Section \ref{sec:magic}, we prove Theorem \ref{thm:1}. Then in Section \ref{sec:separation}, we prove Theorem \ref{thm:2}. In Section \ref{sec:Chern}, we construct the flat band eigenfunctions assuming the simplicity of magic angles and prove Theorem \ref{thm:3}.

\subsection*{Acknowledgement}
The author is very grateful to Maciej Zworski for introducing him to the field of magic angles as well as for having many helpful discussions and comments on the early manuscript. The author would also like to thank Simon Becker for helpful discussions and providing various numerical supports, and Zhongkai Tao, Jie Wang and Jared Wunsch for helpful discussions related to this project and comments on the manuscript.

\section{Hamiltonians of Twisted Multilayer Graphene}
\label{sec:TML}
In this section, we study symmetries of the Hamiltonian $\Hn$. The symmetry group commutes with the Hamiltonian and split the functional space into irreducible representations. This gives rise to the so-called protected states. We also consider eigenstates with zero eigenvalue for $\bt=0$, for which some of them are unprotected in the sense that the corresponding eigenbranches separate from the flat band as we turning on the inter-layer parameter $\bt$ when $\alpha\in\mathcal{A}$. 

\subsection{Symmetries of the Hamiltonian}
\subsubsection{Translational symmetry}
The first identity in \eqref{eq:symmU} shows that for 
$ L_{\mathbf a } \mathbf v ( z ) := \mathbf v ( z + \mathbf a ) $, 
\begin{equation}
  \begin{split}
    \nonumber
    D_n(\alpha; \bt) L_{ \mathbf a } 
    = \mathrm{diag}_n\left\{\begin{pmatrix}  \omega & 0 \\ 0 & 1 \end{pmatrix}  \right\}
    L_{\mathbf a } D_n(\alpha;t)\  \mathrm{diag}_n\left\{ \begin{pmatrix} \bar  \omega & 0 \\ 0 & 1 \end{pmatrix} \right\},
    \ \ \mathbf a = \tfrac{ 4}3 \pi i \omega^\ell , \ \ \ell = 1, 2
  \end{split}
\end{equation}
where $\mathrm{diag}_n\{(*)\}$ denotes block diagonal matrices with $n$ blocks. Therefore, 
\begin{equation}
\label{eq:defLa0} 
\mathscr L_{\mathbf a } D ( \alpha )  = D ( \alpha )  \mathscr L_{\mathbf a } , \ \ \ 
 {\mathscr L}_{\mathbf a } :=  \diag_n\left\{\begin{pmatrix}  \omega & 0 \\ 0 & 1 \end{pmatrix}  \right\} L_{\mathbf a } , \ \ \ 
\mathbf a = \tfrac{ 4}3 \pi i \omega^\ell , \ \ \ell = 1,2. 
\end{equation}
Therefore for the lattice
\begin{equation}
\label{eq:defGam3}
 \Lambda = \tfrac43 \pi i \omega ( \ZZ \oplus \omega \ZZ)  
\end{equation}
and
\begin{equation}
\label{eq:defLa} 
\mathscr L_{\mathbf a} := \diag_n\left\{\begin{pmatrix}  \omega^{a_1 + a_2} & 0 \\ 0 & 1 \end{pmatrix} \right\}   L_{\mathbf a} , \ \ 
\mathbf a = \tfrac 43 \pi i \omega ( a_1 + \omega a_2 ) \in \Lambda, 
\end{equation}
we obtain a unitary action of $\Lambda$ on $ L^2 ( \CC; \CC^{2n} ) $, or more precisely on 
$L^2 ( \CC/3\Lambda; \CC^{2n} )$, with $\mathbf a \mapsto \mathscr L_{\mathbf a }$. We can further extend the action of $ \mathscr L_{\mathbf a } $ to 
$ L^2 ( \CC ; \CC^{4n} ) $ or $ L^2 ( \CC/3\Lambda; \CC^{4n} ) $ block-diagonally
and it yields 
\begin{equation}
\mathscr L_{\mathbf a} H_n ( \alpha ) = H_n ( \alpha ) \mathscr L_{\mathbf a }, \ \ 
\mathbf a  \in \Lambda.
\end{equation}

\subsubsection{Rotational symmetry}
The second identity in \eqref{eq:symmU} shows that
$$  [ D ( \alpha ) \mathbf u ( \omega \bullet ) ] ( z )  = \bar \omega 
[ D ( \alpha ) \mathbf u ] ( \omega z ) .$$
Therefore, applying this to $D_n(\alpha, \bt)$ in equation \eqref{eq:defD}, we have 
\begin{gather}
[ D_n ( \alpha; \bt ) \mathbf u ( \omega \bullet ) ] ( z )  = \bar \omega 
[ J D_n ( \alpha; \bt ) J^* \mathbf u ] ( \omega z ),\\
  J := \diag(1,1, \bar{\omega}, \omega, \bar{\omega}^2, \omega^2 \cdots, \bar{\omega}^{n-1}, \omega^{n-1}).
\end{gather}
Similarly for $D^*_n ( \alpha; \bt )$, we have 
$$  [ D^*_n ( \alpha; \bt ) \mathbf u ( \omega \bullet ) ] ( z )  =  \omega 
[ J D^*_n ( \alpha; \bt ) J^* \mathbf u ] ( \omega z ).$$ 
Hence, for $\mathbf u \in  L^2 ( \CC ; \CC^{4n} )$,
\begin{gather}
\label{eq:CanD} 
\mathscr C' \begin{pmatrix} J 
&  \\  &  J \end{pmatrix} \Hn \begin{pmatrix} J^* 
&  \\  &  J^* \end{pmatrix} = \Hn \mathscr C', \\
\mathscr C' \mathbf u ( z ) := \begin{pmatrix} \text{Id}_{2n} &  \\ 0 &  \bar \omega \text{Id}_{2n}  \end{pmatrix} \mathbf u ( \omega z ).
\end{gather}
Therefore, we obtain
\begin{gather}
\mathscr C \Hn = \Hn \mathscr C, \ \ 
\mathscr{C} := \mathscr C' \begin{pmatrix} J 
&  \\  &  J \end{pmatrix}
\end{gather}

\subsubsection{Additional symmetries}
We record some additional actions and symmetries involving $ \Hn$. 
\begin{equation}
\label{eq:defW}
\begin{gathered}   \Hn = - \mathscr W \Hn \mathscr W^*, \ \ \ \mathscr W := \begin{pmatrix}
1 & 0 \\
0 & -1 \end{pmatrix} ,
\ \ \mathscr W  \mathscr C  = 
 \mathscr C  \mathscr W ,
 \ \ \mathscr L_{\mathbf a } \mathscr W = \mathscr W \mathscr L_{\mathbf a },
\end{gathered}
\end{equation}
This implies that the spectrum of $\Hn$ is even. 
\begin{rem}
    With the notation introduced in \eqref{eq:Dba} we also have 
    $$H_{n,\kk}(\alpha; \bt)  = - \mathscr W H_{n,\kk}(\alpha; \bt) \mathscr W^*.$$
    This implies that for each eigenvalue $E_j(\kk)\geq0$ of $H_{n,\kk}(\alpha; \bt)$, there exists another eigenvalue $E_{-j}(\kk)\leq0$ such that $E_{-j}(\kk) = -E_{j}(\kk)$.
\end{rem}

\subsubsection{Group actions on functional spaces}
Following \cite{becker2020mathematics}, since $ \mathscr C \mathscr L_{\mathbf a } = \mathscr L_{\bar \omega \mathbf a } \mathscr C$, we can combine the two actions into a unitary group action that
commutes with $\Hn$:
\begin{equation}
\label{eq:defG}  
\begin{gathered}  G := \Lambda \rtimes \ZZ_3 , \ \ 
  \ZZ_3 \ni k : \mathbf a \to \bar \omega^k  \mathbf a , \ \ \ ( \mathbf a , k ) \cdot ( \mathbf a' , \ell ) = 
( \mathbf a + \bar \omega^k \mathbf a' , k + \ell ) ,
\\  ( \mathbf a, \ell ) \cdot \mathbf u = 
\mathscr L_{\mathbf a } \mathscr C^\ell \mathbf u . \ \
\end{gathered}
\end{equation}
Taking a quotient by $ 3\Lambda $, we obtain a finite group acting 
unitarily on $ L^2 ( \CC/3\Lambda; \CC^{4n} ) $ and commuting with $ \Hn $:
\begin{equation}
\label{eq:defG3}
G_3 :=  G/3\Lambda = \Lambda/3\Lambda \rtimes \ZZ_3 \simeq \ZZ_3^2 \rtimes \ZZ_3. 
\end{equation}
Restricting to the components, $G$ (resp.\ $G_3$) acts on 
$ L^2 ( \CC ; \CC^{2n} ) $  (resp.\ $ L^2 ( \CC/3\Lambda; \CC^{2n} ) $) as well and we
use the same notation for those actions. With the previous discussions, we have the following
\begin{prop}
  \label{prop:specH}
The operator $ \Hn: L^2 ( \CC; \CC^{4n} ) \to L^2 ( \CC; \CC^{4n} )$ is
an unbounded self-adjoint operator with the domain given by 
$ H^1 ( \CC; \CC^{4n} ) $. The operator $ \Hn $ commutes with the 
unitary action of the group $ G $ given by \eqref{eq:defG} and 
\[ \Spec_{ L^2 ( \CC ) } \Hn  = 
- \Spec_{ L^2 ( \CC ) } \Hn . \]
The same conclusions are valid when $ L^2 (\CC ) $ is replaced by $ L^2 (\CC/3\Lambda) $
and $ G $ by $ G_3 $ given by \eqref{eq:defG3}. In that case, the spectrum is discrete. 
\end{prop}

\subsection{Protected states at zero}
\label{sec:protect}
To show the existence of protected zero eigenstates of the operator $\Hn$, we start with considering all irriducible representations of $G$ (resp.\ $G_3$) on $L^2(\CC)$ (resp.\ $ L^2 (\CC/3\Lambda)$).

\subsubsection{Irreducible representations}
Recall as in \cite[Section 2.2]{becker2020mathematics}, irreducible unitary representations of $ \ZZ_3^2 $ are one dimensional 
and are given by 
\begin{equation}
\label{eq:reprZ32}
\begin{gathered}   \pi_{\mathbf k} : \ZZ_3^2 \to {\mathsf U}(1) , \ \ \  \pi_{\bf k }(\mathbf a )  = e^{i\langle \mathbf a, \mathbf k \rangle} ,\\ 
\mathbf a = \tfrac 43 \pi i \omega  ( a_1 + a_2 \omega ) , \ \  \
a_j \in \ZZ_3 , \ \ \ 
\mathbf k = \tfrac{1}{\sqrt 3 } \omega ( \omega k_1 - k_2 )  , \ k_j 
\in \ZZ_3, \\ 
 \langle \mathbf a, \mathbf k \rangle = \tfrac 1 2 ( \mathbf a \bar {\mathbf k } + 
\bar {\mathbf a} \mathbf k ) =
\tfrac{2 \pi}3 (k_1 a_1 + k_2 a_2) .
\end{gathered}
\end{equation}
Irreducible representations of $ G_3 $ are one dimensional for 
$ \mathbf k \in \Delta $ (given by $ \Delta(\ZZ_3 ) := \{ ( k, k ) , 
 k \in \ZZ_3 \}$ -- note that $ \langle \mathbf k , 
 \omega \mathbf a \rangle = \langle \mathbf k , \mathbf a \rangle$, 
 $ \mathbf a \in \Lambda/3\Lambda $, 
  if and
 only if $ \mathbf k \in \Delta  $), 
\[  \rho_{ k,p } ( ( \mathbf a , \ell ) ) = 
\bar \omega^{\ell p}  \pi_{(k,k)} ( \mathbf a ) ,\]
or three dimensional, for $ \mathbf k \notin \Delta  $:
\[   \rho_{\mathbf k}( ( \mathbf a , \ell ) ) 
= \begin{pmatrix} \omega^{ \langle \mathbf k , \mathbf a \rangle }
& 0 & 0 \\
0 & \omega^{\langle  \mathbf k , \omega \mathbf a \rangle }
 & 0 \\
0 & 0 &  \omega^{\langle  \mathbf k , \omega^2 \mathbf a \rangle }

\end{pmatrix} \begin{pmatrix} 0 & 1 & 0 \\
0 & 0 & 1 \\
1 & 0 & 0 \end{pmatrix}^{\ell}  \in {\mathsf U }( 3 ) . \]
The representations are equivalent for $ \mathbf k $ in the same orbit of
the transpose of $ \mathbf a \mapsto \omega \mathbf a $,
and hence there are only two. 

Defining the following functional space:
\begin{equation}
L^2_{k,p}:= \left\{\mathbf{u}\in L^2(\CC/3\Lambda; \CC^{4n}): \  \mathscr L_{\mathbf a } \mathscr C^{\ell} \mathbf u = e^{i\langle \mathbf a, \mathbf k \rangle} \bar \omega^{\ell p}   u\right\}
\end{equation}
There are then 11 irreducible representations: 9 one dimensional and 2 three dimensional. We can decompose 
$ L^2 ( \CC/3\Lambda; \CC^{4n} ) $ into 11 orthogonal subspaces:
\begin{equation}
  \label{eq:L2-decomp}
  L^2 ( \CC/3\Lambda; \CC^{4n} ) = \bigg(\bigoplus_{k, p \in 
  \ZZ_3 }  L^2_{ {k,p}  }\bigg)  
  \oplus L^2_{  {(1,0)} }  
  \oplus L^2_{  {(2,0)} }.
\end{equation} 
In view of Proposition \ref{prop:specH} we define
\begin{equation}
  H_{{k,p}}(\alpha; \textbf{t}) := \Hn   : L^2_{ {k,p} }\cap H^1 (\CC / 3\Lambda; \CC^{4n}) \longrightarrow L^2_{ {k,p} },
\end{equation}
with $ H_{{(1,0)}} $ and $ H_{{(0,1)}} $ similarly defined. Here we omit the dimensional constant $n$ when there is no ambiguity.

\subsubsection{Protected zero eigenstates}
We start with the case $\alpha=0, \bt\in (\RR\backslash\{0\})^{n-1}$. The kernel of $\Hn$ at $\alpha=0$ is given by 
$$\ker_{L^2 ( \CC/3\Lambda ; \CC^{4n}) } H_n(0;\bt)= \{\mathbf{e}_1, \mathbf{e}_{2n}, \mathbf{e}_{2n+2}, \mathbf{e}_{4n-1}\},$$ 
where the set $ \{\mathbf{e}_j\}_{j=1}^{4n} $ forms the standard basis of $ \CC^{4n} $. Now we consider the action of $G_3 = \ZZ_3^2 {\rtimes} \ZZ_3$ on $\ker_{L^2 ( \CC/3\Lambda; \CC^{4n}) } H_n(0; \textbf{t})$ with $\mathbf a = \frac 4 3 \pi i \omega ( a_1 + a_2 \omega ), \ (a_1,a_2)\in \ZZ_3^2$:
\begin{gather*}  
  \mathscr L_{\mathbf a } \mathbf e_1 =  \omega^{a_1 + a_2 } \mathbf e_1 , \ \ \ \mathscr L_{\mathbf a } \mathbf e_{2n} =  \mathbf e_{2n} ,\ 
  \ \ \mathscr L_{\mathbf a } \mathbf e_{2n+2} = \mathbf e_{2n+2}, \ \ \
  \mathscr L_{\mathbf a } \mathbf e_{4n-1} = \omega^{a_1 + a_2 }  \mathbf e_{4n-1},
  \\ \mathscr C  \mathbf e_1 = \mathbf e_1 , \ \ \
  \mathscr C  \mathbf e_{2n} = \bar \omega^{1-n} \mathbf e_{2n} , \ \ \
  \mathscr C  \mathbf e_{2n+2} = \bar \omega \mathbf e_{2n+2}, \ \ \ 
  \mathscr C  \mathbf e_{4n-1} = \bar \omega^n \mathbf e_{4n-1}.
\end{gather*}

If $n\neq 3k$ for $k\in \mathbb{N}$, then $\mathbf{e}_1, \mathbf{e}_{2n}, \mathbf{e}_{2n+2}, \mathbf{e}_{4n-1}$ are in different irriducible representations of $L^2 ( \CC/3\Lambda; \CC^{4n})$ under group action $G_3$ with 
\[ \mathbf e_1 \in L^2_{ {1,0}} , \ \ \ 
\mathbf e_{2n} \in L^2_{ {0,[1-n]}} , \ \ \ 
\mathbf e_{2n+2} \in L^2_{ {0,1}}, \ \ \ 
\mathbf e_{4n-1} \in L^2_{ {1,[n]}}, 
\] 
where $[k]\in \{0,1,2\}$ with $[k]\equiv k\mod 3$. 
Therefore if $n\neq 3k$ and $\alpha =0$, each of these irriducible representation $ H_{{k,p}} (0; \textbf{t}) $ has a single zero eigenvalue. Since $ \mathscr W $ (see \eqref{eq:defW}) commutes with the action of $ G_3 $, the spectra of $ H_{{k,p}} ( \alpha; \bt ) $ are symmetric with respect to $0$ (see Proposition \ref{prop:specH}), it follows that for each $k, p$ as above, $ H_{{k,p}} ( \alpha; \bt ) $ has an eigenvalue at $0$. 

When $n=3k$ for some $k\in \mathbb{N}$, we have
\[ \mathbf e_1, \mathbf e_{4n-1} \in L^2_{ {1,0}} , \ \ \ 
\mathbf e_{2n}, \mathbf e_{2n+2} \in L^2_{ {0,1}}.
\] 
since $\mathbf e_{4n-1}$ and $\mathbf e_{2n}$ are always eigenfuctions of $H_{{1,0}} ( \alpha; \bt )$ and $H_{{0,1}} ( \alpha; \bt )$ correspondingly, they are protected zero eigenstates naturally. The symmetry of the spectra of $H_{{1,0}}$ (resp.\ $H_{{0,1}}$) implies that $H_{{1,0}} ( \alpha; \bt )$ and $H_{{0,1}} ( \alpha; \bt )$ both have two protected eigenvalues at $0$.
Since 
\begin{equation}
\label{eq:kerH-decomp}
    \begin{split}
        \ker_{L^2 ( \CC/3\Lambda ; \CC^{4n} ) } \Hn = 
 \ker_{L^2 ( \CC/3\Lambda ; \CC^{2n} ) } &\Dn \oplus \{ 0_{\CC^{2n} } \} \\
  + \{ 0_{\CC^{2n} } \} \oplus &\ker_{L^2 ( \CC/3\Lambda ; \CC^{2n} ) } D^*_n(\alpha,\bt),
    \end{split}
\end{equation}
we obtained the following result about symmetric protected eigenstates at $0$:
\begin{prop}
\label{prop:prot}
For all $ \alpha \in \CC $ and any $\bt\in (\RR\backslash \{ 0 \})^{n-1}$ 
\[ \ker_{ L^2_{{ 1,0} } }  \Dn  \neq \emptyset , \ \ 
\ker_{ L^2_{{ 0,[1-n]} }}  \Dn =
\{ \mathbf{e}_{2n} \}.
\]
\end{prop}
\begin{rem}
  The protected states are also discussed in \cite{wang2022hierarchy} using perturbation theory for $\bt$ infinitesimal. Here we provide a different approach to show the zero eigenstates are protected for all $\bt \in \RR^{n-1}$.
\end{rem}

\subsubsection{Unprotected zero eigenstates and band separation}
Now we consider the case $\textbf{t}=0$, i.e., no tunneling in the upper and lower $n$-layers. We show that in this case, there are additional $4n-4$ zero eigenstates of $H_n(\alpha; 0)$ that is unprotected, in the sense that they result in band separations as we turning on the tunneling parameter $\textbf{t}$ (cf. Section \ref{sec:separation}.).

The operator $D_n(\alpha; 0)$ is decomposed into direct sums of one twisted bilayer operator $D(\alpha)$ (see equation \eqref{eq:defD1}) and $D(0)= \mathrm{diag}{(2D_{\bar z}, 2D_{\bar z})}$. The kernel of $H_n(\alpha; 0)$ is $4n$-dimensional and given by direct sums of $\ker_{L^2(\CC/3\Lambda; \CC^{2})} D(\alpha)$ and $\ker_{L^2(\CC/3\Lambda; \CC^{2})} D^*(\alpha)$ and $(2n-2)$ copies of $\ker_{L^2(\CC/3\Lambda; \CC^{2})} D(0)$:
\begin{equation*}
  \ker_{L^2 ( \CC/3\Lambda; \CC^{4n}) } H_n(\alpha;0)= \{\Phi_i(\alpha)\}_{i=1,2,2n+1, 2n+2}\cup \{\mathbf{e}_{j}\}_{j\neq 1,2, 2n+1, 2n+2}
\end{equation*}
where
\begin{gather*}
    \Phi_i(\alpha) = (\varphi_i(\alpha), 0_{\CC^{4n-2}}), \ \ 
    \Phi_{2n+i} (\alpha)= (0_{\CC^{2n}}, \psi_i(\alpha), 0_{\CC^{2n-2}}),\ \ i=1,2;\\ 
    D(\alpha)\varphi_i=0, \ \ D^*(\alpha)\psi_i = 0, \ \ i=1,2.
\end{gather*}
We can again consider the action of $G_3 = \ZZ_3^2 {\rtimes} \ZZ_3$ on $\ker_{L^2 ( \CC/3\Lambda; \CC^{4n}) } H_n(\alpha; 0)$ with $\mathbf a = \frac 4 3 \pi i \omega ( a_1 + a_2 \omega ),\ (a_1, a_2)\in \ZZ_3^2$. 

For simplicity of presentation, we only discuss the case $n=2$. This yields 
\begin{gather}
  \label{eq:unpro1} 
\Phi_1 \in L^2_{ {1,0}}, \ \ \ 
\mathbf{e}_{4} \in L^2_{ {0,2}},\ \ \  
\Phi_{6} \in L^2_{ {0,1}}, \ \ \ 
\mathbf{e}_{7} \in L^2_{ {1,2}},\\
    \label{eq:unpro2} 
\{\Phi_2, \mathbf{e}_{8}\} \subset L^2_{ {0,0}}, \ \ \ 
\{ \mathbf e_3, \Phi_{5}\} \subset L^2_{ {1,1}}.
\end{gather} 
Similar to previous discussion, for $t=0$, each of these irreducible representation $ H_{{k,\ell}} (\alpha; 0)\,(k=0,1;\, \ell= 0,2)$ has a single zero eigenvalue, whereas $H_{{\ell,\ell}} (\alpha; 0) \,(\ell = 0,1)$ has a double zero eigenvalue. The commutativity of $ \mathscr W $ with $ G_3 $ yields that $ H_{{k,\ell}} ( \alpha )$ each have an eigenvalue at $0$. Whereas $\ker H_{{\ell,\ell}} (\alpha; t) = \emptyset$ for $\ell = 0,1$ from the discussion in the previous subsection shows that the pairs of zero eigenfunctions $\{\Phi_2, \mathbf{e}_{8}\}$ and $\{ \mathbf e_3, \Phi_{5}\}$ of $H_{{\ell,\ell}} (\alpha; 0)$ disappears as we turning on the tunneling parameter $t$. This, in fact, is the origin of the Dirac band separation from the flatband, which is discussed in Section \ref{sec:separation}. 

Same argument applies to any $n\in\NN$, where all the irreducible representations $ H_{{k,\ell}} (\alpha; 0)$ may have more than one zero eigenfunctions. The only protected zero eigenstates are eigenbranchs corresponding to eigenfunctions $\{\Phi_1, \mathbf{e}_{2n}, \Phi_{2n+2}, \mathbf{e}_{4n-1}\}$.

\section{Bloch--Floquet theory}
\label{sec:floquet}
In this section we introduce a Bloch--Floquet theory of the Hamiltonian $\Hn$ with respect to the translation operator $\mathscr L_{\mathbf{a}}, \ \mathbf{a}\in\Lambda$, which gives the band structure of the operator $H_{n,\kk}(\alpha; \bt)$ (cf.\,equation \eqref{eq:Dba}) for $\kk\in \CC/\Lambda^*$.

\subsection{A generalized Floquet theory approach}
\label{sec:flo2}
We follow \cite{becker2022fine} for the discussion of Floquet theory with respect to the operator $\mathscr L_{\mathbf a }$ with $\mathbf a \in \Lambda $, which has already implicitly appeared in the physics literature (cf.\,\cite{tarnopolsky2019origin}). One of the advantages of this type of Floquet theory is that the set $\ker_{{ L^2 ( \CC/3\Lambda ; \CC^{2} ) }} D(\alpha)$ at the flat band, which is nine dimensional (cf.\,\cite[Section 3.3]{becker2022fine}), splits to elements of $\ker_{{ L_{\mathbf 0}^2 ( \CC/\Lambda ; \CC^{2} ) }} (D(\alpha)+\mathbf k)$ (see \eqref{eq:L0} below for the definition of $L_{\mathbf 0}^2$) for nine different $\kk\in \CC/\Lambda^*$. This generalized Floquet theory also simplifies the presentation of band separation mechanism (cf.\, Section \ref{sec:separation}). 

We first define the inner product on $\CC$ by 
$$\langle z, w\rangle := \frac{1}{2}(z\bar w+\bar z w),\ \ z,w\in \CC.$$
Note that if $z\in \Lambda$ and $w\in\Lambda^*$, then we have $\langle z, w\rangle= 2\pi \ZZ$.

The generalized Floquet condition is formulated for $\mathbf w \in H^1_{\rm{loc} } (\mathbb C; \mathbb C^{4n})$ as follows:
\begin{equation}
\label{eq:flo}   
\mathscr L_{\mathbf a } \mathbf w ( z ) :
= e^{    i \langle \mathbf a , \mathbf k \rangle }
\mathbf w   ( z ) ,  \ \ \ \mathbf a \in \Lambda, \ \ \mathbf k \in \mathbb C/ \Lambda^*  , \end{equation}
where
\begin{equation}
\label{eq:akbracket}
\begin{gathered} \langle \mathbf a , \mathbf k \rangle =2\pi  (a_1 k_1 + a_2 k_2 )/3  , \ \ \
\mathbf a = \tfrac 4 3 \pi i \omega ( a_1 + \omega a_2 ) \in \Lambda , \ \ a_j \in \mathbb Z , \\
\mathbf k = \frac{1}{\sqrt{3}}\omega( \omega k_1 - k_2 ) \in \CC/\Lambda^* , \ \ 0\leq  k_j < 3.
\end{gathered}
\end{equation}
We consider the Floquet spectrum (cf.\,Proposition \ref{p:slight} for the discreteness of the spectrum) with Floquet condition \eqref{eq:flo} under the twisted translation operator $\mathscr L_{\mathbf a}$:
\begin{gather}
  \label{eq:floquetbc2}
  \Hn \mathbf w_j ( \alpha, \mathbf k ) = E_j ( \alpha, \kk; \bt ) 
\mathbf w_j ( \alpha, \mathbf k ),\\ 
\mathscr L_{\mathbf a } \mathbf w_j ( \alpha, \mathbf k ) = e^{  i \langle \mathbf a , \mathbf k \rangle }
\mathbf w_j ( \alpha, \mathbf k ) , \ \
\mathbf a \in \Lambda, \ \ \mathbf k \in \mathbb C/ \Lambda^*.
\end{gather}
We define
\begin{equation}
  \label{eq:floquet-phy}
  {\mathbf v}_j ( \alpha , \mathbf k ) := e^{  -i \langle z, \mathbf k \rangle } \mathbf  w_j
( \alpha, \mathbf k ) ( z ) , \ \ \ 
\mathbf a \in  \Lambda , \ \ \mathbf k \in \mathbb C/ \Lambda^*,
\end{equation}
which yields $\mathscr L_{\mathbf a } {\mathbf v}_j ( \alpha, \mathbf k ) = 
{\mathbf v}_j ( \alpha, \mathbf k )$. We then consider the operator
\begin{equation}
\label{eq:Dba}   
H_{n, \mathbf k } ( \alpha; \bt ) := e^{-i\langle z, \mathbf k \rangle}  \Hn e^{i\langle z, \mathbf k \rangle} 
= \begin{pmatrix}  0 & \Dn^* + \bar {\mathbf k} \\
\Dn + {\mathbf k } & 0 \end{pmatrix}.
\end{equation}
with boundary condition $ \mathscr L_{\mathbf a } \mathbf v = \mathbf v $ and that gives
 us {\em equivalent} Floquet spectrum:
\begin{equation}
\begin{split}
\label{eq:new_Floquet}
 H_{\mathbf k, n } ( \alpha;\bt )  {\mathbf v}_j ( \alpha , \mathbf k ) =  E_j ( \alpha, \mathbf k, \bt ) 
 {\mathbf v}_j ( \alpha, \mathbf k ) , \ \ 
\mathscr L_{\mathbf a }  {\mathbf v}_j ( \alpha, \mathbf k ) =  
{\mathbf v}_j ( \alpha, \mathbf k ) , \ \ \mathbf a \in \Lambda.
\end{split}
\end{equation}
Note that ${\mathbf v}_j ( \alpha, \mathbf k )$ also depends on $\bt$ but we omit it here for simplicity. This suggests that we consider $ H_{n,\mathbf k } ( \alpha; \bt ) $ as a self-adjoint operator on 
\begin{equation}
\label{eq:L0}
L^2_{\mathbf 0} = L^2_{\mathbf 0} ( \mathbb C/3\Lambda ; \mathbb C^{4n} ) := \{ \mathbf v \in L^2_{\rm{loc}} ( 
  \mathbb C/3\Lambda, \mathbb C^{4n} ) : \mathscr L_{\mathbf a} \mathbf v = \mathbf v , \ \mathbf a \in \Lambda  \}
\end{equation}
with the domain given by $ H^1_{\mathbf 0 } :=  L^2_{\mathbf 0 }  \cap H^1 (\mathbb C/3\Lambda ; \mathbb C^{4n})$. Note that we sometimes slightly abuse the notation of $L^2_{\mathbf 0}$ by restricting to the subspace $\CC^{2}$ or $\CC^{2n}$ of $\CC^{4n}$ the operator $\mathscr{L}_{\mathbf{a}}$ is diagonal.

In other words, we can define a generalized Bloch transform (cf.\,\cite[Chapter 5]{zworskipde}): 
\begin{gather}
    \mathcal{B} u(z,\kk):= \sum_{\mathbf{a}\in\Lambda} e^{-i\langle z+\mathbf{a}, \kk\rangle}\mathscr{L}_{\mathbf{a}}u(z)
\end{gather}
which can be checked easily satisfies
\begin{gather*}
    \mathcal{B} u(z,\kk+\mathbf{p}) = e^{-i\langle z, \mathbf{p}\rangle} \mathcal{B} u(z,\kk), \ \ \mathbf{p}\in \Lambda^*,\\
    \mathscr{L}_{\mathbf{a}'} \mathcal{B} u(\bullet,\kk) =  \sum_{\mathbf{a}\in\Lambda} e^{-i\langle z+\mathbf{a}+\mathbf{a}', \kk\rangle}\mathscr{L}_{\mathbf{a}+\mathbf{a}'}u(z) = \mathcal{B} u(\bullet,\kk),\ \ \mathbf{a}'\in \Lambda,
\end{gather*}
so that $\mathcal{B} u(\bullet,\kk)\in L^2_{\mathbf{0}}$. By the commutativity of $\Hn$ and $\mathscr{L}_{\mathbf{a}}$ we have
\begin{equation*}
       \mathcal{B} \Hn  =  H_{n,\kk}(\alpha;\bt) \mathcal{B}.
\end{equation*}
For a fixed $ \mathbf k \in \CC / \Lambda^* $, 
$ H_{n,\kk}= \mathcal B H_n \mathcal B^* $ acts on $L^2_{\bo}$ 
as the operator in \eqref{eq:Dba}. The modified Bloch transform thus gives rise to bands of eigenfunctions $E_j( \alpha, \mathbf k, \bt ),\ j\in\NN$ of $H_{n,\kk}(\alpha;\bt)$ for $\kk\in \Lambda^*$.

\begin{rem}
 In $\CC/\Lambda^*$, we denote two special points $\mathbf k=(0,0) = \mathbf 0$ by $K$ and $\mathbf k = (1,1)= \frac{1}{\sqrt{3}}\omega( \omega - 1 ) =-\mathbf i$ by $K'$ as in physics literature. 
\end{rem}


\subsection{Floquet theory and representations}
\label{sec:flo-rep}
We further define the spaces (cf.\,the decomposition \eqref{eq:L2-decomp}): 
\begin{equation}
  \label{eq:L2p}  
  \begin{gathered} 
  L^2_{\mathbf p}  :=
   \{ \mathbf u \in 
  L^2 ( \CC/ 3\Lambda ; \CC^{4n} ) :   \mathscr L_{\mathbf a } \mathbf u =
  e^{  i \langle \mathbf a , \mathbf p \rangle } \mathbf u, \ \ \forall \,  \mathbf a \in \Lambda \},
  \\
  L^2 ( \mathbb C/3\Lambda; \mathbb C^{4n} ) = \bigoplus_{ \mathbf p \in 
  3\Lambda^*/\Lambda^* }  L^2_{\mathbf p} .
  \end{gathered}
\end{equation}
Note that for
\[ \mathbf p \in \Delta =  \{ \tfrac 1 {\sqrt 3 } \omega ( \omega p  - p  ) : p \in \mathbb Z_3 \}, \]
we have orthogonal decomposition
\[   L_{\mathbf p}^2  = 
\bigoplus_{ \ell \in \mathbb Z_3 } L^2_{ {p, \ell} }, \ \ \
\mathbf p =  \tfrac 1 {\sqrt 3 } \omega ( \omega p  - p  ) .\]

\begin{rem}
\label{rem:ker_equiv}
    Note that $u\in \ker_{L^2_{\bo}}(\Dn+\mathbf{p})$ if and only if for 
    $v:= e^{i\langle z,\mathbf{p}\rangle}u$, $v\in \ker_{L^2_{\mathbf{p}}}\Dn$.
\end{rem}

Comparing the remaining terms in the decomposition \eqref{eq:L2-decomp} and \eqref{eq:L2p} we obtain
\[  L^2_{  {(j,0)} }
= \bigoplus_{ \mathbf p \in \mathscr O_j } L^2_{\mathbf p} ,
\ \ \ j = 1, 2 ,\]
where $\mathscr O_j,\ j=1,2$ being two orbits under action $\mathbf p \mapsto \bar \omega \mathbf p$ given by
\begin{equation}
  \label{eq:defO}   \mathbb Z_3 \setminus \Delta = \mathscr O_1 \sqcup \mathscr O_2 = 
  \{ (1,0) , (0,2), (2,1) \} \sqcup \{ (2,0), (0,1), (1,2) \} . 
\end{equation}

\subsection{Spectral characterization of magic angles}
We start by introducing some basic properties of the operator $ \Dn $ with domain $H^1_{\mathbf{0}} $.
We first observe that 
\begin{equation}
\label{eq:spec0}
\begin{split}
\Spec_{{ L^2_{\mathbf{0}}} } \Dnn = \Lambda^*+\{\mathbf{0}, \mathbf{i}\},
\end{split} 
\end{equation}
since for $e_{\mathbf k} ( z) := e^{  -i \langle z, \kk \rangle }$, 
\begin{gather*}
    (\Dnn + \kk )  e_{\mathbf{\kk} } \mathbf e_1 = 0, \ \ 
 \mathbf k \in \Lambda^*+\{-\mathbf{i}\};\ \ 
    (\Dnn+\kk)  e_{\mathbf k } \mathbf e_{2n} = 0, \ \ 
 \mathbf k \in \Lambda^*+\{\mathbf{0}\},
\end{gather*}
where the exponentials $ e_\mathbf k / \vol( \CC/\Lambda)^{\frac12},\ \kk\in\Lambda^* $ form an orthonormal basis of $ L^2 ( \CC/\Lambda ) $ and $ \{\mathbf e_j\}_{1\leq j\leq 2n} $ is the standard basis of $ \CC^{2n} $. This leads to the following proposition:

\begin{prop}
\label{p:slight}
The family $ \CC \ni \alpha \mapsto \Dn : H^1_{\bo}  \to L^2_{\bo}$
is a holomorphic family of elliptic Fredholm operators of index $ 0 $,
and for all $ \alpha $, the spectrum of $\Dn$ is $\Lambda^*$-periodic:
\begin{equation}
\label{eq:per}
\Spec_{ L^2_{\bo} } \Dn = 
\Spec_{ L^2_{\bo} } \Dn + \mathbf k , \ \ 
\mathbf k \in \Lambda^*. 
\end{equation}
\end{prop}

\begin{proof}  
  Since $ D_{\bar z } $ is an elliptic operator in dimension 2,
  existence of parametrices (cf. \cite[Proposition E.32]{dyatlov2019mathematical}) implies the Fredholm property (cf. 
  \cite[\S C.2]{dyatlov2019mathematical} for Fredholm operators). In view of \eqref{eq:spec0},  $  \Dnn + \mathbf k $
  is invertible for  $ \kk \notin \Lambda^*+\{\bo, -\bi\} $ and hence
   $ \Dnn : H^1_{\bo}  
  \to L^2_{\bo} $ is an operator of
  index $ 0 $. The same is true for the Fredholm family
  $ \Dn $.
  
  To see \eqref{eq:per}, note that 
  if $ ( \Dn + \lambda ) \mathbf{u} = 0 $ with $\mathbf u \in H^1_{\bo}$,
  then $ ( \Dn + \lambda + \mathbf k ) ( e_{\mathbf k} \mathbf u)  = 0 $ with $e_{\mathbf k} \mathbf u \in H^1_{\bo} $ for any $ \mathbf k \in \Lambda^* $.
\end{proof}

Therefore for each $ \mathbf k $, the operator $ H_{n,\mathbf k} ( \alpha; \bt ) $
is an elliptic differential system and hence it has a discrete spectrum that then describes the spectrum of $ \Hn $ on $ L^2 ( \CC ) $:
\begin{equation}
\label{eq:specH} 
\begin{gathered}    \Spec_{L^2 ( \CC) }   \Hn  = \bigcup_{ \mathbf k \in \CC / \Lambda^* }
\Spec_{L^2_{\bo} }  \Hkn , \\ 
\Spec_{L^2_{\bo} } \Hkn = \{ \pm E_j 
( {\mathbf k } , \alpha, \bt ) \}_{ j=1}^\infty, \ \  E_{ j+1 } ( \mathbf k, \alpha, \bt ) \geq E_j ( \mathbf k, \alpha, \bt ) 
\geq 0.
\end{gathered} 
\end{equation}
Note that sometimes we omit $(\alpha, \bt)$ and simply write $H_{n,\mathbf k}, D_n$ and $E_j(\mathbf{k})$ when there is no ambiguity.
To conclude the second identity in \eqref{eq:specH}, note that
\begin{equation}
\label{eq:H-resolvent}
    (\lambda - \mathscr A )^{-1} = \begin{pmatrix}  ( \lambda^2 - A^* A )^{-1} & 0 
\\ 0 & ( \lambda^2 - A A^* )^{-1} \end{pmatrix} \begin{pmatrix}
 \lambda & A^* \\
 A & \lambda \end{pmatrix}, \ \  \mathscr A := \begin{pmatrix}
 0 & A^* \\ A & 0 \end{pmatrix}.
\end{equation}
Hence, non-zero eigenvalues of $ H_{n,\kk} $ on $L^2_{\bo}$ are given by 
$ \pm $ the non-zero singular values
of $ D_{n} + \mathbf k $, i.e., eigenvalues of the operator
$ [ ( D^*_n + \bar {\mathbf{k}} ) ( D_n + \mathbf k) ]^{\frac12} $, counted with their multiplicities. 

The zero eigenvalue (if exists) of 
$ ( D^*_n + \bar{\mathbf k} ) ( D_n + \mathbf k)$ also has the 
same multiplicity as the zero eigenvalue of
$ ( D_n + \mathbf k ) ( D^*_n + \bar{\mathbf k}) $, 
so that eigenvalues $ E_j (\mathbf k) = 0$ are included
exactly twice (for $ \pm $), this follows from the identity
\begin{equation}
  \label{eq:ker-Dn}
  \dim \ker_{{ L^2_{\bo}  }} (D_n + \mathbf k)  = \dim \ker_{
{ L^2_{\bo} }}  (D_n^* + \bar{\mathbf k}) ,
\end{equation}
which follows from Proposition \ref{p:slight}: the operator 
$ D_n + \mathbf k $ is a Fredholm operator of index $0$.

Combining with the equation \eqref{eq:Dba}, The above discussion yields three equivalent characterizations of the existence of a \emph{flat band} at energy zero:
\begin{prop}
\label{p:flat}
The following statements are equivalent for $\alpha\in \CC$: 
\begin{enumerate}
  \item  $0 \in \bigcap_{ \mathbf k \in \CC } \Spec_{{ L^2_{\bo}}}  H_{n,\kk}(\alpha;\bt)$
  \item $E_0 ( \mathbf k , \alpha, \bt ) = 0 \text{ for all $ \mathbf k \in \CC/\Lambda^* $.}$
  \item $\Spec_{{ L^2_{\bo} }}  D_{n}( \alpha; \bt )= \CC$.
\end{enumerate}
\end{prop}

\section{Magic angles of the Twisted multilayer Hamiltonian}
\label{sec:magic}
The {\em magic angles} are defined as the angles $ 1/\alpha $'s at which 
\begin{equation}
\label{eq:defmal}
0 \in \bigcap_{ \mathbf k \in \CC } 
\Spec_{ L^2_{\mathbf{0}} } H_{n,\mathbf k } ( \alpha; \bt )
\end{equation}
As discussed in Section \ref{sec:floquet}, the Hamiltonian $H_{n,\mathbf k} ( \alpha;\bt )$ comes from the {\em Floquet theory}
of $\Hn$ and \eqref{eq:defmal} means that $ \Hn $ has a
{\em flat band} at $0$ which is given by the equivalent condition:
\begin{equation}
  \label{eq:defmal2}
  0 \in \bigcap_{ \mathbf k \in \CC } \Spec_{ L^2_{\mathbf 0}}  H_{n,\mathbf k } ( \alpha; \bt )  \ 
  \Longleftrightarrow \ 
  \Spec_{ L^2_{\mathbf 0} }  \Dn = \CC.
\end{equation}

\subsection{Magic angles of TBG}
\label{sec:TBG}
We briefly recall several results on flat bands and magic angles of TBG from \cite{becker2022fine}.

Let $\mathcal A$ denote the sets of all magic angles of $D(\alpha)$ with $\alpha\in \CC$, then the spectrum of $D(\alpha)$ on $ L^2_{\mathbf 0 }$ is given by
\[    \Spec_{ L^2_{\mathbf 0 } }
D ( \alpha ) = \left\{ \begin{array}{ll} \Lambda^*   +  \{  \mathbf 0 , \mathbf i \} & 
\alpha \notin \mathcal A , \\
\ \ \ \ \mathbb C & \alpha \in \mathcal A. 
\end{array} \right. \]
This yields the following proposition: 
\begin{prop}
\label{prop:bands}
Suppose $ \alpha \in \mathbb C $ and $ E_1 ( \alpha, \mathbf k ) $ is 
defined using \eqref{eq:floquetbc2} for $ H(\alpha)$ given 
by \eqref{eq:defH} for $n=1$. Then 
\begin{equation}  
\label{eq:alphac1}  
\exists \, \ \mathbf k \notin \Lambda^* + \{ \mathbf 0, -\bi \}  \ \ \  E_1 ( \alpha, \mathbf k ) =  0 
\ \Longleftrightarrow \ \forall \, \ \mathbf k \in \mathbb C  \ \ \  E_1 ( \alpha, \mathbf k ) =  0 . 
\end{equation}
In other words, zero energy band is flat if and and if the Bloch eigenvalue is $ 0 $ at 
some $ \mathbf k \notin \Lambda^* + \{ \mathbf 0, - \mathbf i  \} $, 
which is the lattice of Dirac points $K, K'$.
\end{prop}
The following result is useful when constructing flat band eigenfunctions using the theta function argument when the multiplicity of the flat band is one. 
\begin{prop}
\label{prop:zs}
Suppose that $ \alpha \in \mathcal A $.
If  $   \dim \ker_{L^2_{\mathbf 0} } D( \alpha ) = 1 $, then 
$ \mathbf u \in  \ker_{L^2_{\bf 0} } D ( \alpha ) $, $
 \mathbf u \not \equiv 0 $, has zero of order one at $ - z_S + \Lambda $ and no other zeros. 
\end{prop}
In the above proposition, the point $z_{S}=\frac{1}{3}(\frac{4}{3} \pi i\omega ( -1+ \omega ))$ is a point of high symmetry (cf.\,\cite{becker2020mathematics}). Recall that $\Lambda = \frac{4}{3} \pi i\omega ( a_1+ \omega a_2 )$, so $z_S=\frac{1}{3}(-1,1)\in \frac{1}{3}\Lambda$.

\subsection{Magic angles of TMG}
In this section, we prove the following theorem, which states that magic angles of $\Hn$ coincides with magic angles of TBG.
\begin{theo}
  \label{thm:idd}
   For any $\alpha\in \CC$ and $D(\alpha)$ defined in equation \eqref{eq:defD1}, we have
   \begin{equation}
    \Spec_{ L^2_{\mathbf 0} }  \Dn = \Spec_{ L^2_{\mathbf 0} }  D(\alpha)
    = \left\{ \begin{array}{ll} \Lambda^*   +  \{  \mathbf 0 , \mathbf i \} & 
\alpha \notin \mathcal A , \\
\ \ \ \ \mathbb C & \alpha \in \mathcal A. 
\end{array} \right.
   \end{equation} 
\end{theo}

\begin{proof}
By Proposition \ref{prop:prot} and \ref{p:slight}, for any $\alpha\in\CC$ both $\Spec_{ L^2_{\bo} }  \Dn$ and $\Spec_{ L^2_{\bo} }  D(\alpha)$ contains $\Lambda^*+\{\bo, \bi\}$ as the spectrum and they are both $\Lambda^*$-periodic for spectrum of $D(\alpha)$). Therefore we only consider $\mathbf k\notin \Lambda^*+\{\bo, \bi\}$. 

We define the following matrices
\begin{gather*}
     J_{\bt,+}  
     := \left(\begin{matrix} 0 & t_1T_+R_0(\kk) & \\  & 0 & t_2T_+ R_0(\kk) & \\ &  & 0 & \ddots \\ && \ddots && \\ &&&& t_{n-1}T_+ R_0(\kk) \\ &&&  & 0 \end{matrix}\right),
\end{gather*}
\begin{gather*}
    J_{\bt,-} 
    := \left(\begin{matrix} 0 & & \\ t_1T_- R_0(\kk) & 0 & & \\ & t_2T_- R_0(\kk) & 0 & \ddots \\ && \ddots && \\ &&&& \\ &&& t_{n-1}T_- R_0(\kk) & 0 \end{matrix}\right),
\end{gather*}
where $R_0(\kk):= (D(0)-\kk)^{-1}: L^2_{\bo} \rightarrow H^1_{\bo}$ for $\kk\notin \Lambda^*+\{\bo,\bi\}$.
We only need to show that for any $\alpha\in \CC$ fixed $\Dn-\mathbf k$ is invertible if and only if $D ( \alpha ) - \mathbf k$ is invertible for $\kk\notin \Lambda^*+\{\bo,\bi\}$. Following from a direct computation and the identity $T_+ R_0(\kk) T_-=0$, we have the following decomposition
\begin{equation}
\label{eq:Dn_decomp}
  \begin{split}
    D_n( \alpha; \mathbf{t} ) - \mathbf k
    =&( \id + J_{\bt,+} ) 
    \diag(D(\alpha)-\kk, D(0)-\kk, \cdots, D(0)-\kk) (\id + J_{\bt,-}).
  \end{split}
\end{equation}
Since $\Spec_{ L^2_{\mathbf 0}}  D(0)= \Lambda^*+\{\bo,\bi\}$, the operator $\Dn-\mathbf k$ is invertible if and only if $D ( \alpha ) - \mathbf k$ is invertible for $\kk\notin \Lambda^*+\{\bo,\bi\}$.
\end{proof}

\begin{rem}
    The decomposition \eqref{eq:Dn_decomp} also concludes that the multiplicity of the flat band of $\Dn$ is the same as the multiplicity of the flat band of $D(\alpha)$.
\end{rem}

\section{Band separation with inter-layer tunneling}
\label{sec:separation}
We discuss, using the perturbation theory, the band separation mechanism appeared in \cite{wang2022hierarchy} once the tunneling $\mathbf{t}$ being turned on. We set up a Grushin problem and compute the separation of the band explicitly. For a detailed introduction on the Grushin problem, we refer to \cite[Appendix C]{dyatlov2019mathematical}. 

\begin{figure}
  \includegraphics[height=5.5cm]{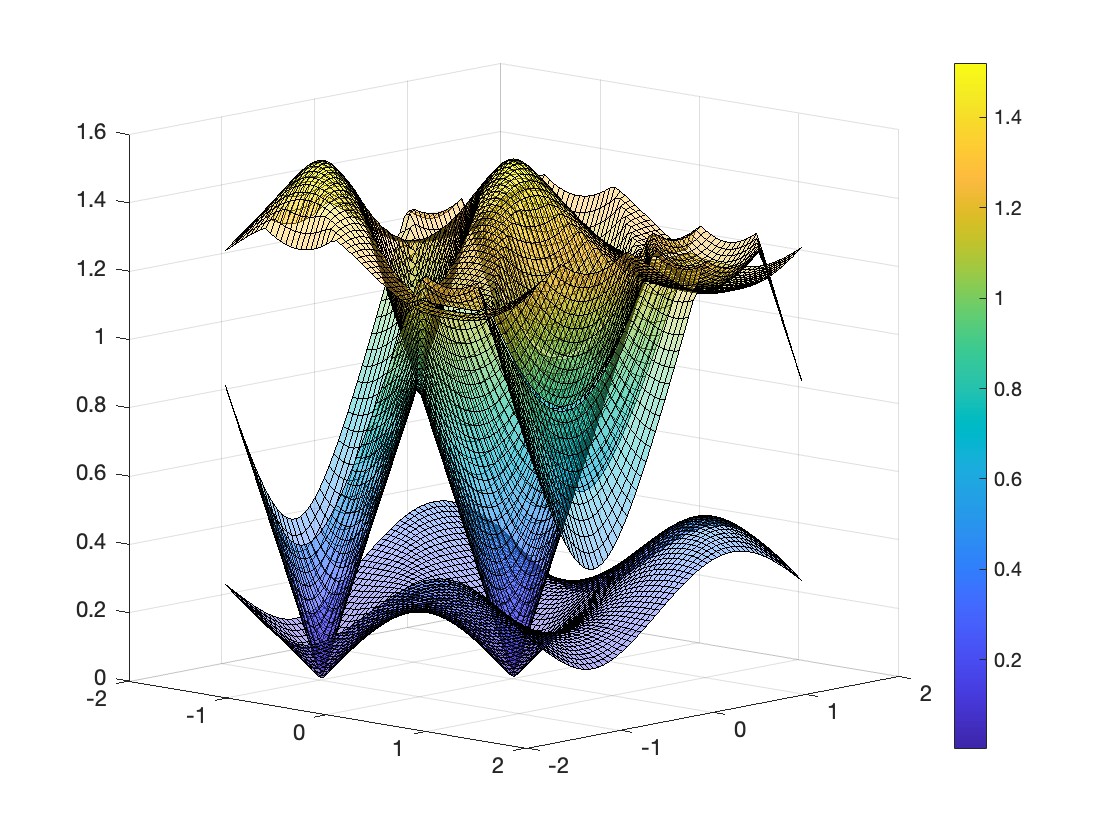}
  \includegraphics[height=5.5cm]{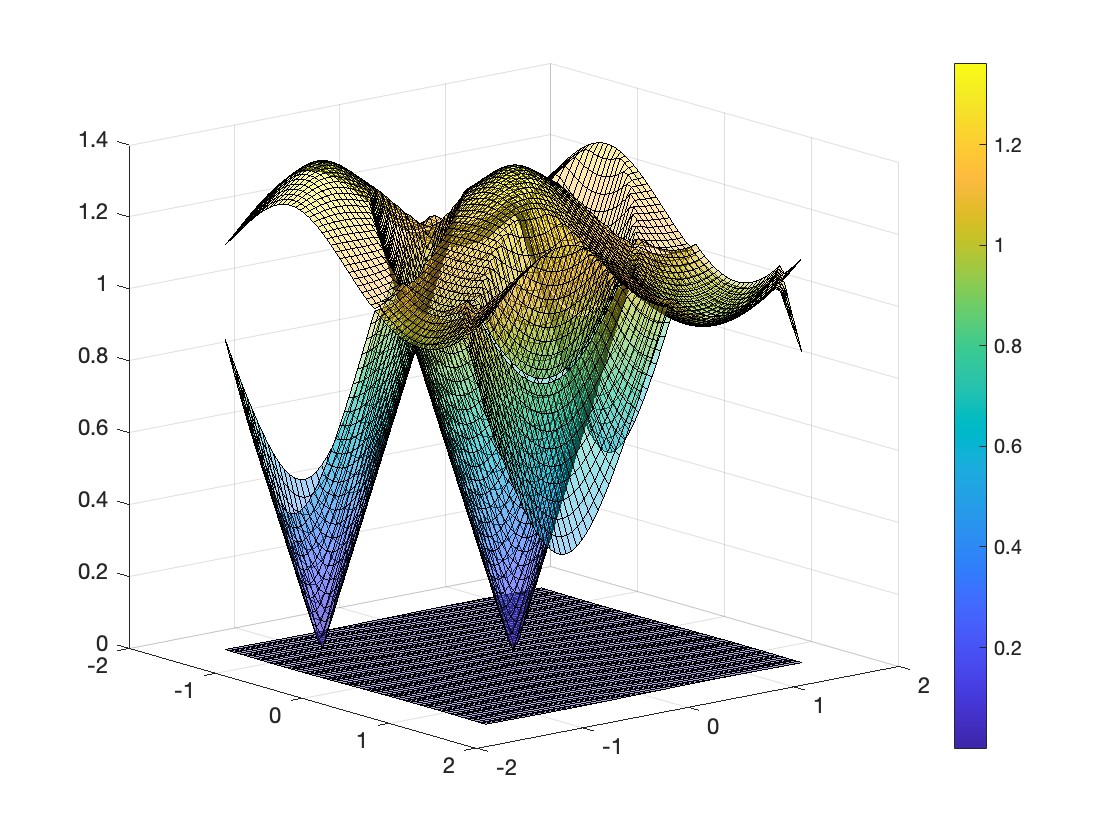}
  \includegraphics[height=5.5cm]{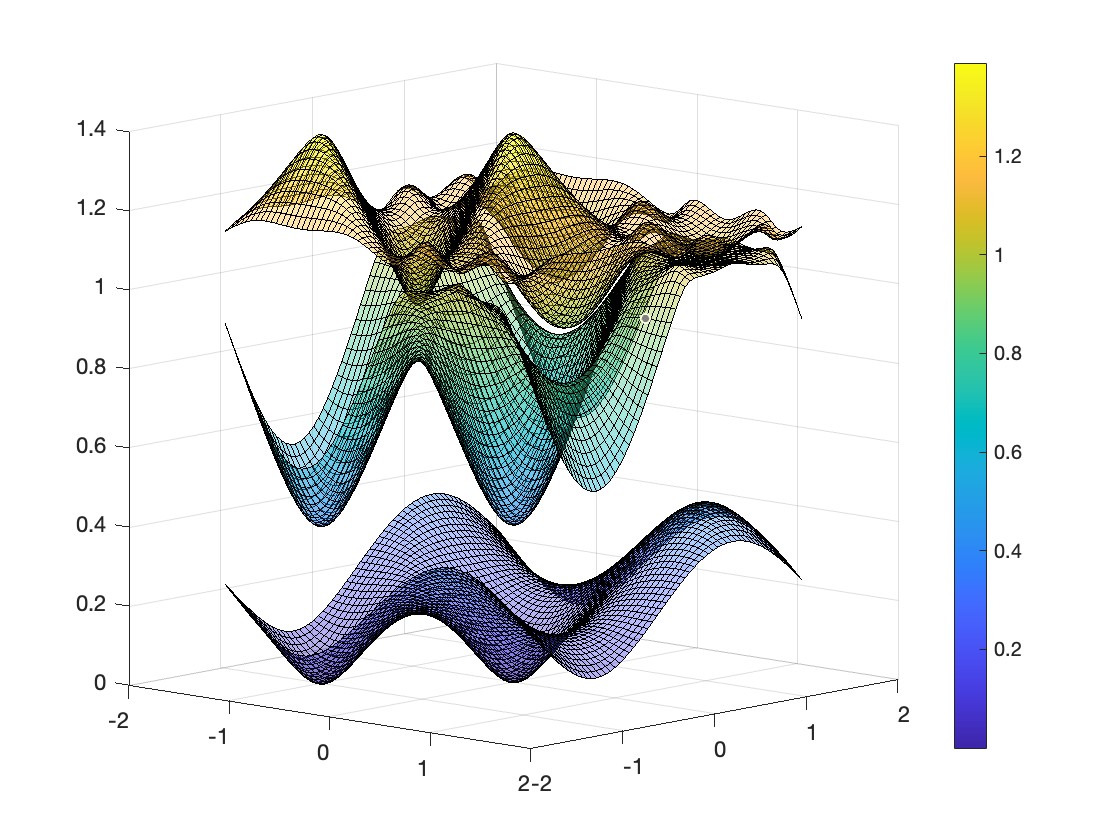}
  \includegraphics[height=5.5cm]{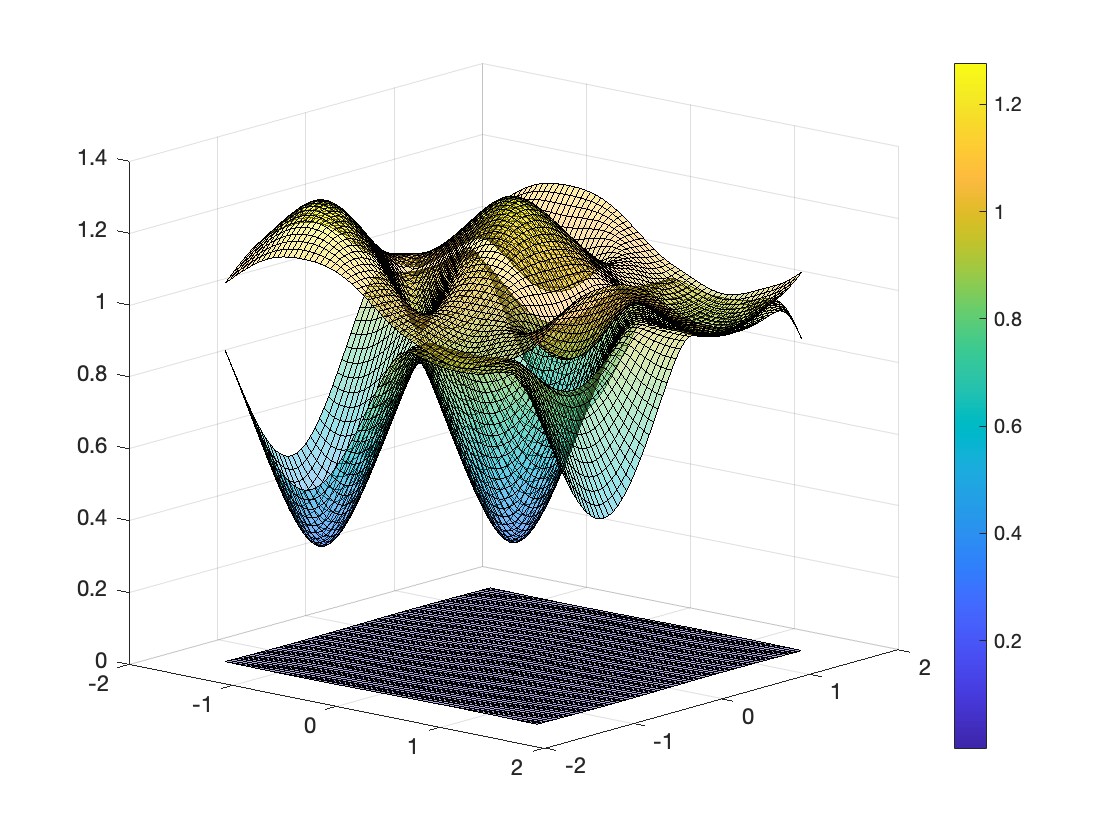}
  \caption{First three band structure $\kk\mapsto E_j(\alpha,\kk, t)$ of $H_2(\alpha; t)$ (cf.\,equation \eqref{eq:new_Floquet}) away from (left) and at the first magic angle of $H(\alpha)$ (right). Top row: $t=0$; there exist two Dirac cones at points $K,K'$. Bottom row: $t=0.5$; the Dirac bands separate from the flat band and a band gap is open.}
  \label{fig:bandsep}
\end{figure}

\subsection{A Grushin Problem}
Following the identity \eqref{eq:H-resolvent}, eigenvalues of $H_{n,\kk}(\alpha; \bt)$ are given by singular values of $\Dn + \mathbf k $, i.e., eigenvalues of 
$ [ ( D_n^*(\alpha; \bt) + \bar{\mathbf k} ) ( \Dn + \mathbf k) ]^{\frac12} $, counted with multiplicities. For simplicity purposes, we work with the case $n=2$, the general case is essetially the same. We first introduce the Schur complement formula:

\begin{prop}
  \label{prop:Grushin1}
  Suppose 
  \begin{align}\label{2.2.1}
  \begin{pmatrix}
  P&R_-\\
  R_+&R_{+-}
  \end{pmatrix}
  = 
  \begin{pmatrix}
  E&E_+\\
  E_-&E_{-+}
  \end{pmatrix}^{-1}:X_1\times X_-\to X_2\times X_+  
  \end{align}
  are bounded operators on Banach spaces, then $P$ is invertible if and only if $E_{-+}$ is invertible.
  Moreover, in such case we have
  \begin{align}
    \label{Grushin}
      P^{-1}=E-E_+E_{-+}^{-1}E_-,\quad E_{-+}^{-1}=R_{+-}-R_+P^{-1}R_-.
  \end{align}
\end{prop}
\begin{proof}
  See \cite[Appendix C.1]{dyatlov2019mathematical}.
\end{proof}

Define $D_{\kk}(\alpha):= D(\alpha)+\kk$ and $D_{2,\kk}(\alpha; t):= D_2(\alpha; t)+\kk$. We have 
$$D_{2,\kk}(\alpha; t)= D_{2,\kk}(\alpha; 0)+ tT, \ \ \ T= \begin{pmatrix}  
  0 & T_+
  \\ T_- & 0 
\end{pmatrix}.$$
We consider the operator $P^{t}_{\kk}= D_{2,\kk}^*(\alpha; t) D_{2,\kk}(\alpha; t)$ given by 
\begin{equation*}
  \begin{split}
    P^{t}_{\kk}
    =&P^0_{\kk} +tQ +t^2 T^*T\\
    :=&D_{2,\kk}^*(\alpha; 0) D_{2,\kk}(\alpha; 0)+ t [T^*D_{2,\kk}(\alpha; 0)+ D_{2,\kk}^*(\alpha; 0)T ] + t^2 T^*T
  \end{split}
\end{equation*}
We start by considering the zero spectrum of the operator $P^0_{\kk} = D_{2,\kk}^*(\alpha; 0) D_{2,\kk}(\alpha; 0)$. By equation \eqref{eq:unpro1} and \eqref{eq:unpro2} and the discussion of \S \ref{sec:flo-rep}, $\ker D_2(\alpha;0)$ are given by 
\begin{equation}
  \Phi_1, \mathbf e_3\in L^2_{-\mathbf i},\ \ 
  \Phi_2, \mathbf e_4\in L^2_{\mathbf 0}.
\end{equation} 
Here we slightly abuse notations by restricting $\ker H_2(\alpha;0)=\{\Phi_1, \Phi_2, \mathbf e_3, \mathbf e_4\}$ to the first component $\ker D_2(\alpha;0)$ because of the direct sum decomposition \eqref{eq:kerH-decomp}. Following the spirit of Remark \ref{rem:ker_equiv}, we define
\begin{equation}
  \hat \Phi_1:= e^{i\langle z,\mathbf i\rangle} \Phi_1,\ \hat{\mathbf e}_3 := e^{i\langle z,\mathbf i\rangle} \mathbf{e}_3,\ \hat \Phi_2:= \Phi_2,\ \hat{\mathbf e}_4 := \mathbf{e}_4 \in L^2_{\mathbf 0}
\end{equation}
These are elements of $\ker_{L^2_{\mathbf 0}} H_{2,\kk}(\alpha; 0)$, or more precisely $\ker_{L^2_{\mathbf 0}} P^0_{\kk}$, for $\kk= -\mathbf{i}, \mathbf{0}$ respectively.

\subsection{Analysis at K'}
To see how zero spectrum changes as we turning on $t>0$, we consider the following Grushin problem for the perturbed operator $P^t_{\kk}-z$ at the Dirac point $\kk= K' = -\mathbf i$: 
\begin{equation}
  \label{eq:t-Grushin}
  \begin{pmatrix}
    P^t_{\kk}-z & R_-\\
    R_+ & 0
    \end{pmatrix}: H^1_{\mathbf{0}} \oplus \CC^2 \longrightarrow L^2_{\mathbf{0}} \oplus \CC^2
\end{equation}
with 
\begin{gather*}
  R_-: (u^{(1)}_-, u^{(2)}_-)^T \mapsto u^{(1)}_- \hat \Phi_1 + u^{(2)}_- \hat \be_3,\ 
  R_+: u \mapsto (\langle u, \hat \Phi_1\rangle, \langle u, \hat \be_3\rangle)^T. 
\end{gather*}
We use $\mathcal{I}$ to denote the set $\{\hat \Phi_{1}, \hat \be_{3}\}$. For $t=0$, the Grushin problem \eqref{eq:t-Grushin} is invertible with the inverse given by 
\begin{equation}
  \label{eq:Grushin-inverse}
  \mathcal{E} = 
  \begin{pmatrix}
    E & E_+\\
    E_- & E_{-+}
    \end{pmatrix}: L_{\bo}^2 \oplus \CC^2 \longrightarrow H_{\bo}^1 \oplus \CC^2
\end{equation}
with 
\begin{gather*}
    Ev= \sum_{u_j\notin \mathcal{I}}\frac{1}{z_j-z}\langle v, u_j\rangle u_j,\ \ E_+v_+= R_-v_+, \\
    E_-v= R_+v,\ \ 
    E_{-+}= \begin{pmatrix}
      z & \\
       & z
      \end{pmatrix} 
\end{gather*}
where the set $\{u_j\}\subset L^2_{\mathbf 0}$ is the set of all normalized eigenfunctions of $P_{\kk}^0$ with eigenvalues $\{z_j\}$. Note that by definition of $\Dn$ in equation \eqref{eq:defD}, the eigenspace of $P_{\kk}^0=D_{2,\kk}^*(\alpha;0)D_{2,\kk}(\alpha;0)$ is given by the direct sum of eigenspaces of $D_{\kk}^*(\alpha)D_{\kk}(\alpha)$ and $D_{\kk}^*(0)D_{\kk}(0)$. Therefore we can take a complete basis of eigenfunctions $\{u_j\}$ with any element $u_j$ is either of the form $(u^{(1)}_j, 0_{\CC^2})^T$ or $(0_{\CC^2}, u^{(2)}_j)^T$, where $\{u^{(1)}_j\}$ are eigenfunctions of $D_{\kk}^*(\alpha)D_{\kk}(\alpha)$ and $\{u^{(2)}_j\}$ are eigenfunctions of $D_{\kk}^*(0)D_{\kk}(0)$. To work with the perturbed operator $P^t_{\kk}$, we employ the following proposition. 
\begin{prop}
  \label{prop:Grushin2}
  For $|t|\ll 1$, the Grushin problem for $P^t_{\kk}-z$:
  $$\mathcal{P}^t=
    \begin{pmatrix}
      P^t_{\kk}-z & R_-\\
      R_+ & 0
    \end{pmatrix}$$ is invertible. If the inverse $(\mathcal{P}^t)^{-1}$ is given by $\begin{pmatrix}
        E^t & E^t_+\\
        E^t_- & E^t_{-+}
        \end{pmatrix}$
      , then 
      \begin{equation}
        \label{eq:Grushin}
      E^t_{-+}= E_{-+}+ \sum_{k=1}^{\infty} (-t)^k E_- (Q+tT)\big(E(Q+tT)\big)^{k-1}E_+
      \end{equation}
\end{prop}
\begin{proof}
  This follows directly from \cite[Propsition 2.12]{zworskipde}.
\end{proof}
Proposition \ref{prop:Grushin2} yields that
\begin{equation*}
  \begin{split}
    E^t_{-+} & =  
    E_{-+} - t E_-Q E_+ + t^2E_{-}(QEQ-T)E_{+}+\mathcal{O}(t^3), 
  \end{split}
\end{equation*}
with 
\begin{equation}
  \label{eq:grushin-1}
  E_-Q E_+ = 
    \begin{pmatrix}
      \langle  Q \hat \Phi_1, \hat \Phi_1\rangle & \langle  Q \hat \be_3, \hat \Phi_1\rangle \\
      \langle  Q \hat \Phi_1, \hat \be_3\rangle & \langle  Q \hat \be_3, \hat \be_3 \rangle
    \end{pmatrix} 
\end{equation}
where the inner product is taken in the space $L^2_{\bo}(\CC/3\Lambda; \CC^4)$. The expression of $Q$ and the fact that $\hat \Phi_{1}$ and $\hat \be_{3}$ being zero eigenfunctions of $D_{2,\kk}(\alpha; 0)$ at $\kk=K'$ yields that equation \eqref{eq:grushin-1} vanishes.

Now we compute the second order perturbation term in $t$. Writing $\hat{\Phi}_1 = (\varphi_1, 0_{\CC^2})^T$ and $\hat{\be}_3= (0_{\CC^2},e_1)^T$ where $\varphi_1, e_3\in L^2_{\bo}(\CC/3\Lambda; \CC^2)$, we have the following identities
\begin{equation*}
    \begin{gathered}
        T \hat{\Phi}_1= (0_{\CC^2},T_-\varphi_1)^T,\ \  T \hat{\be}_3= (T_+e_1, 0_{\CC^2})^T.
    \end{gathered}
\end{equation*}
The operator $E_-TE_+: \CC^2\rightarrow \CC^2$ is then given by the matrix
\begin{equation}
\label{eq:t2-T}
   \begin{split}
   \begin{pmatrix}
      \langle T\hat{\Phi}_1, T\hat{\Phi}_1 \rangle & \langle T\hat{\Phi}_1, T\hat{\be}_3 \rangle \\
      \langle T\hat{\be}_3, T\hat{\Phi}_1 \rangle  &  \langle T\hat{\be}_3, T\hat{\be}_3 \rangle
    \end{pmatrix} 
    =
    \begin{pmatrix}
      \| T_-\varphi_1\|^2 & 0 \\
      0  &  \| T_+e_1\|^2
    \end{pmatrix},
   \end{split}
\end{equation}
while the operator $E_{-}QEQE_{+}: \CC^2 \rightarrow \CC^2$ is given by the matrix
\begin{equation*}
   \begin{split}
   \sum_{u_j\notin\mathcal{I}}\frac{1}{z_j-z}
    \begin{pmatrix}
      |\langle T \hat{\Phi}_1, D_{2,\kk}(\alpha;0)u_j \rangle|^2 & \langle T \hat{\be}_3, D_{2,\kk}(\alpha;0)u_j \rangle  \langle D_{2,\kk}(\alpha;0)u_j, T \hat{\Phi}_1 \rangle  \\
      \langle T \hat{\Phi}_1, D_{2,\kk}(\alpha;0)u_j \rangle \langle  D_{2,\kk}(\alpha;0)u_j, T \hat{\be}_3 \rangle  &  |\langle T \hat{\be}_3, D_{2,\kk}(\alpha;0)u_j \rangle|^2
    \end{pmatrix}.
   \end{split}
\end{equation*}
Writing $u_j=(u^{(1)}_j, u^{(2)}_j)^T$, the above matrix can be rewritten as
\begin{equation*}
   \begin{split}
   \sum_{u_j\notin\mathcal{I}}\frac{1}{z_j-z}
    \begin{pmatrix}
      |\langle T_- \varphi_1, D_{\kk}(0)u_j^{(2)} \rangle|^2 & \langle T_+e_1, D_{\kk}(\alpha)u_j^{(1)} \rangle \langle  D_{\kk}(0)u_j^{(2)}, T_- \varphi_1\rangle \\
      \langle D_{\kk}(\alpha)u_j^{(1)}, T_+e_1 \rangle \langle T_- \varphi_1,  D_{\kk}(0) u_j^{(2)}\rangle  &  |\langle T_+e_1, D_{\kk}(\alpha)u_j^{(1)} \rangle|^2
    \end{pmatrix}.
   \end{split}
\end{equation*}
By the direct sum decomposition of $P_{\kk}^0$, any eigenfunctions $u_j=(u^{(1)}_j, u^{(2)}_j)^T$ is of the form $(u^{(1)}_j, 0_{\CC^2})^T$ or $(0_{\CC^2}, u^{(2)}_j)^T$, the off-diagonal terms in the above matrix vanishes. It thus reduces to a sum of diagonal matrices
\begin{equation}
\label{eq:t2-term}
   \begin{split}
    \sum_{j}\frac{1}{z^{(2)}_j-z}
    \begin{pmatrix}
      |\langle T_- \varphi_1, D_{\kk}(0)u_j^{(2)} \rangle|^2 & 0 \\
      0  &  0
    \end{pmatrix} + \sum_{i}\frac{1}{z^{(1)}_i-z}
    \begin{pmatrix}
      0 & 0 \\
      0  &  |\langle T_+e_1, D_{\kk}(\alpha)u_i^{(1)} \rangle|^2
    \end{pmatrix}
   \end{split}
\end{equation}
in which $z^{(1)}_i\neq 0$ and $z^{(2)}_j\neq 0$ with
$$
(D^*_{\kk}(\alpha)D_{\kk}(\alpha)-z^{(1)}_i)u_i^{(1)}=0,\ \ 
(D^*_{\kk}(0)D_{\kk}(0)-z^{(2)}_j)u_j^{(2)}= 0.
$$
We consider the first term in the expression \eqref{eq:t2-term}. Note that $\{u^{(2)}_j\}$ are joint eigenfunctions of $D^*_{\kk}(0)D_{\kk}(0)$ and $D_{\kk}(0)$ with spectrum
$$\Spec_{L^2_{\bo}}D_{\kk}(0) = \Lambda^*+ \{\bo, -\bi\} \text{ for } \kk=-\bi$$
Therefore,
$$|\langle T_- \varphi_1, D_{\kk}(0)u_j^{(2)} \rangle|^2= z^{(2)}_j |\langle T_- \varphi_1, u_j^{(2)} \rangle|^2.$$
Using the fact that $\langle T_- \varphi_1, e_1 \rangle=0$ and Plancherel theorem, at $z=0$ we obtain the equality
\begin{equation}
\label{eq:t2-term1}
    \sum_{j}\frac{z^{(2)}_j}{z^{(2)}_j-z}
    \begin{pmatrix}
      |\langle T_- \varphi_1, u_j^{(2)} \rangle|^2 & 0 \\
      0  &  0
    \end{pmatrix}
    =
    \begin{pmatrix}
      \|T_- \varphi_1\|^2 & 0 \\
      0  &  0
    \end{pmatrix}   
\end{equation}
For $z=0$, combining equations \eqref{eq:t2-T} \eqref{eq:t2-term} and \eqref{eq:t2-term1}, the second order perturbation is given by
\begin{equation*}
   \begin{split}
   E_{-}(QEQ-T)E_{+} = 
    \begin{pmatrix}
      0 & 0 \\
      0  &  -\|T_+e_1\|^2+\sum_{i}\frac{1}{z^{(1)}_i}|\langle T_+e_1, D_{\kk}(\alpha)u_i^{(1)} \rangle|^2
    \end{pmatrix},
   \end{split}
\end{equation*}
recall that $\{z^{(1)}_i\}$ is the set of non-zero eigenvalues of $D^*_{\kk}(\alpha)D_{\kk}(\alpha)$ with $L^2_{\bo}$-eigenfunctions $\{u^{(1)}_i\}$. Therefore by Proposition \ref{prop:Grushin1}, \ref{prop:Grushin2}, zero is an eigenvalue for $P_{\kk}^t$ closed to $\kk=K'$. The zero eigenvalue corresponds to the flat band of $H_2(\alpha;t)$.

The eigenvalue of the operator $P^t_{\kk}$ corresponding the Dirac cone of $P^0_{\kk}$ is given by $$\lambda_1(t)=t^2\big[\|T_+e_1\|^2_{L^2}-\sum_{i}\frac{1}{z^{(1)}_i}|\langle T_+e_1, D_{\kk}(\alpha)u_i^{(1)} \rangle|^2\big]+\mathcal{O}(t^3),$$ 
which corresponds to the separation of the Dirac band from the flat band. Since this eigenvalue is given by the solution to the equation
$$z-t^2\Big[\|T_+e_1\|^2_{L^2}-\sum_{i}\frac{1}{z^{(1)}_i-z}|\langle T_+e_1, D_{\kk}(\alpha)u_i^{(1)} \rangle|^2\Big]=0,$$
a substitution $z=\lambda t^2$ and a geometric series argument yield the solution $z=\lambda_1(t)$.

Combining with the fact that there is a protected states at $0$ for any $t\in\RR$, this proves the following proposition:
\begin{prop}
\label{prop:sep}
    For $t\ll 1$ and at $\kk=K'$, the first two eigenvalues of $H_{2,\kk}(\alpha;t)$ are given by $E_{1,\kk}(t)=0$ and $E_{2,\kk}(t)=Ct +\mathcal{O}(t^{3/2})$, where the constant is given by
    $$C= \bigg(1-\sum_{i}\frac{1}{z^{(1)}_i}|\langle T_+e_1, D_{\kk}(\alpha)u_i^{(1)} \rangle|^2\bigg)^{1/2}.$$
\end{prop}

\section{Flatband eigenfunctions and Chern numbers}
\label{sec:Chern}
In the section, we construct flat band eigenfunctions of $H_{n,\kk}(\alpha;\bt)$ at $E_{1}(\kk,\alpha;\bt)=0$ for $\kk\in\CC/\Lambda^*$ using a theta function argument analogous to \cite{becker2020mathematics} (see also \cite{ledwith2020fractional}). The construction also gives an explicit method to compute the Chern number of the flat band of this twisted multilayer model. We start by reviewing some backgrounds in theta functions.

\subsection{Review of theta functions}
\label{sec:theta}
We recall that the theta function $\theta ( \zeta ) := - \theta_{\frac12,\frac12} ( \zeta | \omega )$ is given by
\begin{equation}
\label{eq:theta}
\begin{gathered} 
\theta (\zeta) = - \sum_{ n \in \mathbb Z } \exp ( \pi i (n+\tfrac12) ^2 \omega+ 2 \pi i ( n + \tfrac12 ) (\zeta + \tfrac
12 )  ),\ \ \theta(-\zeta)= -\theta(\zeta),\\
\theta ( \zeta + m  ) = (-1)^m \theta  ( \zeta ) , \ \ \theta ( \zeta + n \omega) = (-1)^n e^{ - \pi i n^2 \omega - 
2 \pi i \zeta n } \theta  ( \zeta ),
\end{gathered}
\end{equation}
where
$$\theta_{a,b} ( \zeta | \tau ) := \sum_{ n \in \ZZ } \exp ( 
\pi i ( a  + n)^2 \tau + 2 \pi i ( n + a ) ( \zeta + b ) ), \ \ \Im \tau >0. $$
The function $ \theta (\zeta) $ has simple zeros at 
\[   \mathbb Z \omega + \mathbb Z , \]
see \cite{mumfordtata}.  If we consider 
\[   g( \zeta ) := e^{ 2 \pi i \zeta a  } \frac{ \theta ( \zeta + {k}  ) } { \theta ( 
\zeta  ) } , \]
we obtain a meromorphic function with simple poles  at $  \mathbb Z + \omega \mathbb Z $, 
simple zeros at $  -k + \mathbb Z + \omega \mathbb Z $, 
and 
satisfying
$  g ( \zeta + m + n \omega ) = e^{  2 \pi i ( a m +  ( \omega a -k  ) n ) } g ( \zeta ). $
Hence, by choosing $ k = \omega a - b $, we define
\begin{equation}
\label{eq:defk}  
g_k ( \zeta ) := e^{2 \pi  \zeta (  k -\bar k )/\sqrt 3 }\, \frac{ \theta ( \zeta + k ) }
{\theta ( z ) } ,  \ \ \ k = \omega a - b , \ \ a, b \in \mathbb R  . 
\end{equation}
It satisfies 
\begin{equation}
\label{eq:fper}
\begin{gathered} 
g_k ( \zeta + m + n \omega ) = e^{ 2 \pi i ( a m + b n ) } g_k ( \zeta ), \ \  
k = \omega a - b , \ \ a, b \in \mathbb R , \\
 a m + n b = \tfrac  i { \sqrt 3 }\left( \zeta_0 \bar k - \bar \zeta_0 k \right), \ \ \zeta_0 = m + n \omega ,
\ \  m , n \in \mathbb Z, \\
g_{ k + p \omega - q } ( \zeta ) = (-1)^{p+q} e^{ -2\pi i k p-\pi i p^2\omega  } g_k ( \zeta ) , \ \
p, q \in \mathbb Z . 
\end{gathered} 
\end{equation} 
We then define a function periodic with respect to the lattice $ \omega \mathbb Z  + \mathbb Z $: 
\begin{equation}
\label{eq:F}
    F_k ( \zeta ) := e^{ 2 \pi ( \zeta - \bar \zeta) k /\sqrt 3} \frac{ \theta ( \zeta + k) }{\theta ( \zeta )}.
\end{equation}

\subsection{Construction of flat bands}
\label{sec:theta2}
We first transplant the functions constructed above to the lattice $\Lambda$ and its dual $\Lambda^*$. Recall that in equation \eqref{eq:defGam3}, they are given by
$$\Lambda = \tfrac43 \pi i \omega ( \ZZ \oplus \omega \ZZ), \ \ \Lambda^* = {\sqrt 3 } \omega ( \omega \ZZ \oplus \ZZ).$$
We introduce change of variables 
$$z= \tfrac43 \pi i \omega \zeta, \ \ \kk= {\sqrt 3 } \omega k$$ 
in the equation \eqref{eq:defk} \eqref{eq:F}, so that the new function is periodic with respect to the lattice $\Lambda$, i.e., we define
\begin{equation}
    F_{\kk}(z) :=  F_k \left(\frac{3z}{4 \pi i \omega}\right),\ \ F_{\mathbf k } ( z + \mathbf a ) = F_{\mathbf k } ( z ), \ \mathbf a \in \Lambda,
\end{equation}
Following from \cite[Lemma 3.1]{becker2022fine}, $F_{\kk}(z)$ is then the fundamental solution of $(2D_{\bar z}+\kk)$ in the sense that
\begin{equation}
( 2 D_{\bar z }  + \mathbf k ) F_{\mathbf k } ( z ) = \alpha_{\mathbf k}  \delta_0 ( z ),  \ \
\mathbf \alpha_{\mathbf k} \neq 0, \ \ \mathbf k \notin \Lambda^*,  \ \ 
 z \in \mathbb C/ \Lambda. 
\end{equation}
Therefore the Schwartz kernel of
\begin{equation}
\label{eq:res_ker}
    ( 2 D_{\bar z } + \mathbf k )^{-1} : L^2 (  \mathbb C/ \Lambda )
\to H^1 (  \mathbb C/ \Lambda )
\end{equation}
is given by $c_k F_{\mathbf k } ( z - z' )$, with $c_k = \theta'(0)/(\pi \theta (k)) = {1}/{\alpha_{\kk}}$. 


Recall that in \cite[Propsition 3.4]{becker2020mathematics} for $\alpha\in \mathcal{A}$ simple, a solution to $ ( D ( \alpha ) + \mathbf k ) \mathbf u_{\mathbf k }  
= 0 $ for any $ \mathbf k \in \CC $ is constructed using a theta function argument. In particular, for $ 
\mathbf k \notin \Lambda^*+\{\bo, -\bi\} $ this implies, by Proposition \ref{prop:bands}, that
$ \alpha \in \CC $ is a magic angle of $H(\alpha)$. Recall that the magic angles of $\Hn$ coincides with the magic angles of $H(\alpha)$. We therefore have a similar construction of all flat band eigenfunctions of $H_{n,\kk}(\alpha; \bt)$.

Recall that by Proposition \ref{prop:zs}, for $\alpha\in \mathcal{A}$ simple, i.e. flat band has multiplicity one, $\mathbf{u}\in \ker_{L^2_{\bo}}D(\alpha)$ vanishes at the point $-z_S+\Lambda^*$. Therefore we have
\begin{equation}
 \label{eq:magick}  
 \mathbf v_\mathbf k ( z ) := F_{\mathbf k } ( z + z_S ) \mathbf u ( z ) , \ \ 
 \mathbf v_{\mathbf k } \in  L^2_{\bo}, 
 \ \  
 ( D ( \alpha ) + \mathbf k ) \mathbf v_{\mathbf k } = 0,  
\end{equation}
Define the function
\begin{equation}
    \label{eq:Dn_eigenf}
    \Phi_{\mathbf k}(z) = (\phi_{1,\kk}, \phi_{2,\kk}, \cdots, \phi_{2n,\kk})^T\in L^2_{\bo}
\end{equation}
with
\begin{enumerate}
    \item $\phi_{2j-1,\kk}=0$ for $j=2, 3, \cdots, n$;
    \item $(\phi_{1,\kk}, \phi_{2,\kk})^T= \Theta_{n-1}(\kk)\mathbf{v}_{\kk}(z)$; 
    \item $\phi_{2m+2,\kk} = (\prod_{j=1}^m t_{j}) \cdot (2D_{\bz}+\kk)^{-m} \phi_{2,\kk},\ m=1,2,\cdots, n-1$,
\end{enumerate}
where $(2D_{\bz}+\kk)^{-m}$ denotes the $m$-th powers of the resolvent $(2D_{\bz}+\kk)^{-1}: L^2_{\bo}\rightarrow L^2_{\bo}$ and 
\begin{equation}
    \Theta_{n-1}(\kk):= [\theta(\kk)]^{n-1}
\end{equation}
It is easy to verify that $\Phi_{\mathbf k}(z)$ defined above solves the equation
\begin{equation}
    \label{eq:nmagic}
    (\Dn+\mathbf k)\Phi_{\mathbf k}(z) =0
\end{equation}
for $\kk\notin \Lambda^*$, as the resolvent $(2D_{\bz}+\kk)^{-1}$ is invertible in $\CC\backslash\Lambda^*$. 

For $\kk\in\Lambda^*$, without loss of generality we consider $\kk=\bo$, otherwise multiplication by $e^{-i\langle z,\kk \rangle}\in L^2(\CC/\Lambda)$ reduces it to the zero case. Note that the (Schwartz kernel of) resolvent is meromorphic at $\kk=\bo$ with a simple pole, i.e.,
\begin{equation}
     (2D_{\bz}+\kk)^{-1} = \frac{1\otimes 1}{\kk}+ f_{\kk}(z,z'),
\end{equation}
where $f_{\kk}$ is holomorphic in $\kk$. The function $\Theta_{n-1}(\kk)$ is holomorphic in $\kk$ and has a zero at $\kk=0$ with order $n-1$. Therefore, each component $\phi_{i,\kk}, \ i=1,\cdots, 2n$ is holomorphic in $\kk$ at $\kk=\bo$. In particular, the construction (3) above yields that at $\kk=\bo$
\begin{equation}
\label{eq:intu}
    \phi_{2m+2,\bo} =0,\ m=1,2,\cdots, n-2;\ \ \phi_{2n,\bo}= \bigg(\prod_{j=1}^m t_{j}\bigg) \int_{C/\Lambda} u_2(z)dm(z),
\end{equation} 
where $\mathbf{v}_{\bo}=\mathbf{u}=(u_1,u_2)$. This yields the flat band eigenfunction at $\kk=\bo$ provided that the integral appeared in the above equation is non-zero. This is verified numerically for the first five real magic angles, where the multiplicities of the flat band are one (see Table \ref{tab:int}).

\begin{table}[h!]
    \centering
    \begin{tabular}{ c||c|c|c|c|c|c } 
    $\alpha$ & 0.586 & 2.22 & 3.75 & 5.28 & 6.79 & 8.31 \\
    \hline
    $I(u_2)$ & 0.2345 & 0.0542 & 0.0033 & 0.0022 & 0.0013 & $8.3963\times 10^{-4}$ \\ 
    \end{tabular}
    \caption{Numerical results of the integral $I(u_2)=\int_{C/\Lambda} u_2(z)dm(z)$ in \eqref{eq:intu} of the flat band eigenfunction on one fundamental domain $\CC/\Lambda$}
    \label{tab:int}
\end{table}

\begin{rem}
    This agrees with the protected eigenstates result shown using irreducile representations, where $\mathbf{e}_{2n}\in L^2_{\bo}$ is always a protected eigenstate with zero-energy.
\end{rem}

\subsection{High Chern numbers of the flat band}
\label{sec:chern}
In this section, assuming the flat band is simple, we show that the flat band eigenfunctions constructed above give rise to a holomorphic line bundle with Chern number $c_1(L)=-n$. For basics on holomorphic line bundles, Berry connections and Chern numbers, we refer to \cite{zworskipde}.

We use the notations in \S \ref{sec:theta} for simplicity. By equation \eqref{eq:theta} \eqref{eq:fper} and \eqref{eq:res_ker}, we obtain
\begin{equation}
\label{eq:defep}
\Phi_{k+{p}} = [e^{(n)}_{p}(k)]^{-1}  \tau({p} ) \Phi_{k}
\end{equation}
where
\begin{gather}
\left[ \tau(p) f  \right]( \zeta )  = e^{ - 2 \pi i ( \zeta_1 p_1 + \zeta_2 p_2) } f ( \zeta ) , \ \ \ 
e^{(n)}_{ p } ( k ) = [e^{ \pi i p_1^2 \omega + 2 \pi i k p_1 } (-1)^{p_1 +p_2 }]^n  , \\
  \zeta = \zeta_1 + \omega
\zeta_2 , \ \ p= \omega p_1 - p_2 , \ \ p_j \in \mathbb Z , \ \ k \in \mathbb C . 
\end{gather}
Note that 
\[  \tau ( p ) :  \ker_{L^2_{\mathbf 0 }} ( 
\Dn + k )  \to \ker_{L^2_{\mathbf 0 }} ( 
\Dn + k + p ) , \]
is a unitary transformation. Here we slightly abuse notations by identifying $L^2_{\bo}$ and $\Dn+k$ with corresponding parts in the rescaled lattice space. Following the standard construction we define
\begin{equation}
\label{eq:defL}  
\begin{gathered}
L := \left\{ [ k , \Phi ]_\tau  \in  (\mathbb C \times L^2_{\mathbf 0 }) /
\sim_\tau : \Phi \in \ker_{L^2_{\mathbf 0 }} ( 
\Dn + k ) \right\}, \\ 
[ k , \Phi]_\tau = [ k', \Phi']_\tau \ \Leftrightarrow \
( k , \Phi ) \sim_\tau ( k', \Phi' ) \ \Leftrightarrow \  \exists \, p \in \Lambda  \ 
k' = k + p , \ \  \Phi' = \tau ( p ) \Phi.  
\end{gathered} 
\end{equation}
Therefore, we have the following lemma:
\begin{lemm}
\label{l:holo}
Definition \eqref{eq:defL} gives a holomorphic line bundle over $ \mathbb C / \Lambda^* $:
\[   f: L  \to  \mathbb C/\Lambda^* \  \text{ with }\  [ k , \Phi ]_\tau \mapsto [k] \in \mathbb C / \Lambda^* . \]
The corresponding family of multipliers is given by $ k \mapsto e^{(n)}_p ( k ) $.
\end{lemm}
\begin{proof}
The action of the discrete group $ \Lambda^* $, with 
$ \Lambda^*\ni p : ( k , \Phi ) \mapsto ( k + p , \tau ( p ) \Phi ),$
on the trivial complex line bundle 
\begin{equation}
\label{eq:wL}  \widetilde L := \{ ( k , \kappa \Phi_k ) : k \in \mathbb C , \ \kappa \in \mathbb C \}
\simeq \mathbb C_k \times \mathbb C_\kappa  , \end{equation}
(where $\Phi_{k}$ is defined in \eqref{eq:Dn_eigenf}) is free and proper, and the quotient map is given by 
$ \pi_\tau ( k , \kappa \Phi_k ) = [ k , \kappa \Phi_k ]_{\tau } $. 
Hence its quotient by that action, $ L $, is a smooth complex manifold of dimension 2.

The map $ (p,k ) \mapsto e^{(n)}_p ( k ) $ satisfies the cocycle conditions with the action $\varphi_p$ on $\Tilde{L}$:
\begin{equation}
\label{eq:CD}
\begin{gathered} 
\varphi_p ( k , \kappa  \Phi_k  ) := ( k + p , e^{(n)}_p ( k ) \kappa \Phi_{k+p} ), \ \   
e^{(n)}_{p+p' } ( k ) = e^{(n)}_{ p' } ( k + p ) e^{(n)}_{p}( k) , \ \ 
p, p' \in \Lambda^*,
\end{gathered} 
\end{equation}
and for $ p \in \Lambda^* $. We then have 
$ \pi_\tau ( ( k, \kappa \Phi_k   ) ) = \pi_\tau ( k + p , e^{(n)}_p ( k ) \kappa \Phi_{k+p}   ) $ by equation \eqref{eq:defep} and this 
gives $L$ the structure of a complex line bundle over $ \mathbb C /\Lambda^* $.
\end{proof}

The hermitian structure is inherited from $ L^2 ( \mathbb C / \Lambda) $ and the resulting hermitian 
structure on $ \widetilde L $ of \eqref{eq:wL}. In coordinates $ ( k , \tau ) $ on $ \widetilde L $, 
we get
\[ h ( k ) =    \| \Phi_k  \|_{L^2 ( \mathbb C/3\Lambda)  }^2 ,\]
and this gives a hermitian structure on $L$:
from \eqref{eq:defep}
we see that
\begin{equation}
\label{eq:logh} 
h ( k  ) =  | e^{(n)}_p (k) |^2 h ( k + p  ) , \ \  p \in \mathbb Z \oplus \omega \mathbb Z \cong \Lambda^* .
\end{equation}
We then associate a Chern connection $\theta$
and a Berry curvature $ \Omega = d\theta$ to the hermitian metric $h$. The Chern number of the line bundle $L$ thus can be calculated using the method of multiplier as
\begin{equation}
 \begin{split}
 \label{eq:Chern}
 c_1 (L) 
 & = \frac{i}{2\pi}  \int_{\CC/\Lambda^*} \Omega\\
 & = \frac{i}{2\pi}\left(\log e^{(n)}_{\omega} (1)- \log e^{(n)}_{\omega} (0) - \log e^{(n)}_1 (\omega) - \log e^{(n)}_1 (0)\right)\\
 & = -n. 
\end{split} 
\end{equation}
This concludes the proof to Theorem \ref{thm:3}.

\bibliography{reference}
\bibliographystyle{alpha}
\nocite{*}

\end{document}